%% file: odd-query.tex
\RequirePackage[l2tabu,orthodox]{nag}
\documentclass
[11pt,letterpaper]
{article} 

\usepackage[notes=true,later=false,camera=false]{dtrt}
\usepackage[utf8]{inputenc}
\usepackage{etex}
\usepackage{ stmaryrd }
\usepackage{xspace,enumerate}
\usepackage[T1]{fontenc}
\usepackage[full]{textcomp}
\usepackage[american]{babel}
\usepackage{mathtools}
\usepackage{amsthm}
\usepackage{empheq}

\hypersetup{
colorlinks=true,
urlcolor=Cerulean,
linkcolor=RoyalBlue,
citecolor=OliveGreen,
linktocpage=true,
}
\renewcommand*\backref[1]{\ifx#1\relax \else (pg. #1) \fi}

\usepackage[capitalise,nameinlink]{cleveref}
\crefname{lemma}{Lemma}{Lemmas}
\crefname{fact}{Fact}{Facts}
\newcommand{\colorconstraints}{\text{Color Constraints}}
\crefname{colorconstraints}{(color constraints)}{Color Constraints}
\crefformat{colorconstraints}{#2\colorconstraints#3}
\crefname{indsetconstraints}{(indset constraints)}{IndSet Constraints}
\crefformat{indsetconstraints}{#2$\mathsf{IndSet\ Axioms}$#3}
\crefname{theorem}{Theorem}{Theorems}
\crefname{mtheorem}{Theorem}{Theorems}
\crefname{itheorem}{Theorem}{Theorems}
\crefname{corollary}{Corollary}{Corollaries}
\crefname{claim}{Claim}{Claims}
\crefname{example}{Example}{Examples}
\crefname{algorithm}{Algorithm}{Algorithms}
\crefname{problem}{Problem}{Problems}
\crefname{definition}{Definition}{Definitions}
\crefname{equation}{Eq.}{Eq.}
\crefname{strategy}{Strategy}{Strategies}

\usepackage{paralist}
\usepackage{turnstile}
\usepackage[framemethod=TikZ]{mdframed}
\mdfsetup{frametitlealignment=\center}
\usepackage{tikz}
\usepackage{caption}
\DeclareCaptionType{Algorithm}
\usepackage{newfloat}

\newtheorem{theorem}{Theorem}[section]
\newtheorem{lemma}[theorem]{Lemma}
\newtheorem*{lemma*}{Lemma}
\newtheorem{claim}[theorem]{Claim}
\newtheorem{proposition}[theorem]{Proposition}
\newtheorem{fact}[theorem]{Fact}

\theoremstyle{definition}
\newtheorem{definition}[theorem]{Definition}
\newtheorem*{definition*}{Definition}
\newtheorem{remark}[theorem]{Remark}

\newtheorem{observation}[theorem]{Observation}

\newtheorem{algorithm}{Algorithm}
\newtheorem{algorithm-thm}[theorem]{Algorithm}

\usepackage[
letterpaper,
top=1.2in,
bottom=1.2in,
left=1in,
right=1in]{geometry}
\usepackage{newpxtext} % T1, lining figures in math, osf in text
\usepackage{textcomp} % required for special glyphs
\usepackage[varg,bigdelims]{newpxmath}
\usepackage[scr=rsfso]{mathalfa}% \mathscr is fancier than \mathcal
\usepackage{bm} % load after all math to give access to bold math
\let\mathbb\varmathbb
\usepackage{microtype}

\usepackage{footnotebackref}

\allowdisplaybreaks
\newcommand{\FormatAuthor}[3]{
\begin{tabular}{c}
#1 \\ {\small\texttt{#2}} \\ {\small #3}
\end{tabular}
}

\newcommand{\R}{{\mathbb R}}
\newcommand{\N}{{\mathbb N}}

\newcommand{\eps}{\varepsilon}

\newcommand{\F}{{\mathbb F}}

\newcommand{\E}{{\mathbb E}}
\newcommand{\1}{\mathbf{1}}
\newcommand{\ip}[1]{\langle #1 \rangle}

\newcommand{\C}{\mathbb C}
\newcommand{\Bits}{\{0,1\}}
\newcommand{\zo}{\Bits}

\newcommand{\cH}{\mathcal H}

\newcommand{\poly}{\mathrm{poly}}
\newcommand{\val}{\mathrm{val}}

\newcommand{\mper}{\,.}
\newcommand{\mcom}{\,,}

\newcommand{\Norm}[1]{\left\lVert#1\right\rVert}

\newcommand{\polylog}{\operatorname{polylog}}

\renewcommand{\emptyset}{\varnothing}
\renewcommand{\geq}{\geqslant}
\renewcommand{\leq}{\leqslant}
\renewcommand{\epsilon}{\varepsilon}

\newcommand{\codegree}[2]{d_{#1,#2}}

\usepackage{macros}

\begin{document}
%%%%%%%%%%%%%%%%%%%%%%%%%%%%%%%%%%%%%%%%%%%%%%%%%%%%%%%%%%%%%%%%%%%%%%%%%%%%%%%%
%%%%%%%%%%%%%%%%%%%%%%%%%%%%%%%%%%%%%%%%%%%%%%%%%%%%%%%%%%%%%%%%%%%%%%%%%%%%%%%%
%%%%%%%%%%%%%%%%%%%%%%%%%%%%%%%%%%%%%%%%%%%%%%%%%%%%%%%%%%%%%%%%%%%%%%%%%%%%%%%%

%%%%%%%%%%%%%%%%%%%%%%%%%%%%%%%%%%%%%%%%%%%%%%%%%%%%%%%%%%%%%%%%%%%%%%%%%%%%%%%%
\title{Improved Lower Bounds for all Odd-Query Locally Decodable Codes}

\author{
\begin{tabular}[h!]{ccc}
\FormatAuthor{Arpon Basu}{ab5541@princeton.edu}{Princeton University}
\FormatAuthor{Jun-Ting Hsieh\thanks{Supported by NSF CAREER Award \#2047933.}}{juntingh@cs.cmu.edu}{Carnegie Mellon University} \\ \\
\FormatAuthor{Pravesh K.\ Kothari\thanks{Supported by  NSF CAREER Award \#2047933, Alfred P. Sloan Fellowship and a Google Research Scholar Award.}}{kothari@cs.princeton.edu}{Princeton University}
\FormatAuthor{Andrew D. Lin}{andrewlin@princeton.edu}{Princeton University}
\end{tabular}
} 
\date{November 21, 2024}
%%%%%%%%%%%%%%%%%%%%%%%%%%%%%%%%%%%%%%%%%%%%%%%%%%%%%%%%%%%%%%%%%%%%%%%%%%%%%%%%

%%%%%%%%%%%%%%%%%%%%%%%%%%%%%%%%%%%%%%%%%%%%%%%%%%%%%%%%%%%%%%%%%%%%%%%%%%%%%%%%
\maketitle
%%%%%%%%%%%%%%%%%%%%%%%%%%%%%%%%%%%%%%%%%%%%%%%%%%%%%%%%%%%%%%%%%%%%%%%%%%%%%%%%
\vspace{-1em}

%%%%%%%%%%%%%%%%%%%%%%%%%%%%%%%%%%%%%%%%%%%%%%%%%%%%%%%%%%%%%%%%%%%%%%%%%%%%%%%%
%%%%%%%%%%%%%%%%%%%%%%%%%%%%%%%%%%%%%%%%%%%%%%%%%%%%%%%%%%%%%%%%%%%%%%%%%%%%%%%%
%%%%%%%%%%%%%%%%%%%%%%%%%%%%%%%%%%%%%%%%%%%%%%%%%%%%%%%%%%%%%%%%%%%%%%%%%%%%%%%%
\begin{abstract}
We prove that for every odd $q\geq 3$, any $q$-query binary, possibly non-linear locally decodable  code ($q$-LDC) $E:\{\pm1\}^k \rightarrow \{\pm1\}^n$ must satisfy $k \leq \tilde{O}(n^{1-2/q})$. For even $q$, this bound was established in a sequence of works~\cite{KT00,GKST06,KdW04}. For $q=3$, the above bound was achieved in a recent work~\cite{AlrabiahGKM23} using an argument that crucially exploits known exponential lower bounds for $2$-LDCs. Their strategy hits an inherent bottleneck for $q \geq 5$. 

Our key insight is identifying a general sufficient condition on the hypergraph of local decoding sets called $t$-\emph{approximate strong regularity}. This condition demands that 1) the number of hyperedges containing any given subset of vertices of size $t$ (i.e., its \emph{co-degree}) be equal to the same but arbitrary value $d_t$ up to a multiplicative constant slack, and 2) all other co-degrees be upper-bounded \emph{relative} to $d_t$. This condition significantly generalizes related proposals in prior works~\cite{GuruswamiKM22,HsiehKM23,AlrabiahGKM23,HKMMS24} that demand \emph{absolute} upper bounds on all co-degrees. 

We give an argument based on spectral bounds on \emph{Kikuchi Matrices} that lower bounds the blocklength of any LDC whose local decoding sets satisfy $t$-approximate strong regularity for any $t \leq q$. Crucially, unlike prior works, our argument works despite having no non-trivial absolute upper bound on the co-degrees of any set of vertices. To apply our argument to arbitrary $q$-LDCs, we give a new, greedy, approximate strong regularity decomposition that shows that arbitrary, dense enough hypergraphs can be partitioned (up to a small error) into approximately strongly regular pieces satisfying the required relative bounds on the co-degrees. 

%In contrast to the \emph{weak regularity} decompositions in prior works~\cite{GuruswamiKM22,HsiehKM23} that ensure \emph{upper bounds} on all co-degrees, our decomposition guarantees that the co-degrees of most subsets of size $t$ are multiplicative close to a fixed, but arbitrarily large, $d_t$. 

\end{abstract}

%%%%%%%%%%%%%%%%%%%%%%%%%%%%%%%%%%%%%%%%%%%%%%%%%%%%%%%%%%%%%%%%%%%%%%%%%%%%%%%%
%%%%%%%%%%%%%%%%%%%%%%%%%%%%%%%%%%%%%%%%%%%%%%%%%%%%%%%%%%%%%%%%%%%%%%%%%%%%%%%%
%%%%%%%%%%%%%%%%%%%%%%%%%%%%%%%%%%%%%%%%%%%%%%%%%%%%%%%%%%%%%%%%%%%%%%%%%%%%%%%%

\thispagestyle{empty}
\setcounter{page}{0}

\clearpage
 \microtypesetup{protrusion=false}
  \tableofcontents{}
  \microtypesetup{protrusion=true}
\thispagestyle{empty}
\setcounter{page}{0}

\clearpage

\pagestyle{plain}
\setcounter{page}{1}

%%%%%%%%%%%%%%%%%%%%%%%%%%%%%%%%%%%%%%%%%%%%%%%%%%%%%%%%%%%%%%%%%%%%%%%%%%%%%%%%
%%%%%%%%%%%%%%%%%%%%%%%%%%%%%%%%%%%%%%%%%%%%%%%%%%%%%%%%%%%%%%%%%%%%%%%%%%%%%%%%
%%%%%%%%%%%%%%%%%%%%%%%%%%%%%%%%%%%%%%%%%%%%%%%%%%%%%%%%%%%%%%%%%%%%%%%%%%%%%%%%

\input{intro}
\input{overview}

\input{prelim}

\input{lowerbound}

\input{new_bucketing_lemma}

\section*{Acknowledgments}
We thank Matija Buci{\'c} for helpful discussions on related topics. A.B. thanks Rohit Agarwal for helpful discussions and comments.

%%%%%%%%%%%%%%%%%%%%%%%%%%%%%%%%%%%%%%%%%%%%%%%%%%%%%%%%%%%%%%%%%%%%%%%%%%%%%%%%
%%%%%%%%%%%%%%%%%%%%%%%%%%%%%%%%%%%%%%%%%%%%%%%%%%%%%%%%%%%%%%%%%%%%%%%%%%%%%%%%
%%%%%%%%%%%%%%%%%%%%%%%%%%%%%%%%%%%%%%%%%%%%%%%%%%%%%%%%%%%%%%%%%%%%%%%%%%%%%%%%
\bibliographystyle{alpha}
\bibliography{odd-query}

\end{document}

%% file: intro.tex
%!TEX root = odd-query.tex
%%%%%%%%%%%%%%%%%%%%%%%%%%%%%%%%%%%%%%%%%%%%%%%%
\section{Introduction}
\label{sec:intro}
A binary error-correcting code $E: \on^k \rightarrow \on^n$ is \emph{locally decodable} if it admits a simple ``local'' decoding algorithm. Given a corrupted received word $y \in \on^n$ and an index $1 \leq i \leq k$ of the message, the (randomized) algorithm reads at most $q$ locations of $y$ (where $q$ is the \emph{query complexity}) and is required to output the $i$-th bit of the message correctly with probability $>1/2+\Omega(1)$.
% Such a (randomized) algorithm takes input a corrupted received word $y$ and an index $1 \leq i \leq k$ of the message, reads $y$ in at most $q$ (the \emph{query complexity}) locations, and is required to output the $i$-th bit of the message correctly with probability $>1/2+\Omega(1)$.
The study of locally decodable codes (LDCs) goes back to the early days of algorithmic coding theory. Starting with the proof of the PCP theorem~\cite{AS98,ALMSS98}, their applications form a diverse list that includes worst-case to average-case reductions~\cite{Tre04}, private information retrieval~\cite{Yekhanin10}, secure computation~\cite{IshaiK04}, derandomization~\cite{DvirS05}, matrix rigidity~\cite{Dvir10}, fault-tolerant computation~\cite{Romashchenko06}, and even data structure lower bounds~\cite{Wolf09,ChenGW10}.

The central question in the study of LDCs is the optimal trade-off between the \emph{rate}, i.e., the smallest possible blocklength $n$ as a function of $k$, and the query complexity $q$, assuming that the (normalized) distance of the code is a fixed constant $\delta \in (0,1)$. For $q=2$ (the smallest interesting query complexity), the binary Hadamard code achieves $n = 2^{O(k)}$. This exponential blow-up was eventually shown to be necessary even for non-linear codes over all finite alphabets~\cite{KT00,GKST06,KdW04}.\footnote{For $2$-query codes on alphabet sizes that grow with $n$, known lower bounds get progressively weaker~\cite{WdW05}. For large enough alphabet sizes, this is inevitable as the problem is closely related to the 2-server private information retrieval (PIR) that does admit significant improvements on the parameters of Hadamard codes over large fields~\cite{DG16}.}

For $q\geq 3$, our knowledge is starkly limited. Over 15 years ago, a breakthrough result~\cite{Yek08,Efremenko09} showed that there are $q$-query LDCs ($q$-LDCs, henceforth) with a sub-exponential blocklength of $n \leq \exp\parens{2^{\tilde{O}(\log^{1/(q-1)} k)}} \leq 2^{k^{o(1)}}$. On the other hand, a sequence of results~\cite{GKST06,KdW04}, beginning with the pioneering work of Katz and Trevisan~\cite{KT00}, showed that $k \leq \tilde{O}(n^{1-2/q})$ for any \emph{even} integer $q$. For odd $q$, since $q$-LDCs are $(q+1)$-LDCs, we can invoke the above inequality to yield a weaker bound of $k \leq \tilde{O}(n^{1-2/(q+1)})$. In the following two decades, these bounds were state-of-the-art up to polylogarithmic factor improvements~\cite{Woo07,Woo08,Woo12,BCG20}.

A recent work~\cite{AlrabiahGKM23} introduced a new connection between $q$-LDCs and the problem of semirandom $q$-CSP refutation. By relying on spectral methods based on \emph{Kikuchi Matrices}~\cite{WeinAM19} for semirandom CSP refutation~\cite{GuruswamiKM22,HsiehKM23}, they obtained the first polynomial improvement on the above bounds for the case of $q=3$. In particular, their bound of $k \leq \tilde{O}(n^{1/3})$ matches the expression $\tilde{O}(n^{1-2/q})$ proved in~\cite{KdW04} for even $q$. The argument of~\cite{AlrabiahGKM23} was later simplified in~\cite{HKMMS24} to obtain a purely combinatorial proof for linear $3$-LDCs that also removes all but one logarithmic factor. This framework eventually led to new progress giving exponential lower bounds on the blocklength of $3$-query locally \emph{correctable} codes \cite{KothariM23,KothariM24,Yankovitz24,AlrabiahG24}. 

\parhead{Beyond $q=3$?} Despite the quantitative improvements for $q=3$, the ideas developed in~\cite{AlrabiahGKM23,HKMMS24} fail to give any improvement for any odd $q>3$. Indeed, their arguments strongly exploit properties specific to $q=3$, as we next explain. At a high level, the proof involves a decomposition step that starts from the query sets of a local decoder for a purported $3$-LDC (with $k \gg n^{1/3}$) and outputs either (1) a $3$-LDC where the query sets satisfy an additional pseudo-random property, or (2) an object that is essentially a $2$-LDC with (a significantly reduced) message length of $\sim O(\log n)$. One must obtain a contradiction in either of the two cases to show that $3$-LDCs with $k \gg n^{1/3}$ do not exist. One can only get the $2$-LDC-like object to have message length $\sim \log n$, but this suffices in their setting since $2$-LDCs are known to require an exponential blocklength~\cite{GKST06,KdW04}. The analogous plan for $5$-LDCs would require an exponential lower bound on $q$-LDCs for all $q\leq 4$. But for $q=4$, such a bound is not only far from the best-known quadratic lower bound but, given the Efremenko-Yekhanin codes, also \emph{false}. This inherent bottleneck has so far prevented any improvement in the lower bounds for odd $q\geq 5$. Indeed, upgrading the lower bounds for $q\geq 5$ was explicitly identified as an outstanding open question in prior work~\cite{AlrabiahGKM23}. 

In this work, we resolve this question and prove: 
\begin{theorem}\label{main}
 For any $q$-LDC $E:\on^k \rightarrow \on^n$ with a constant distance,  $k\leq O(n^{1-2/q} \log^4 n)$. If $E$ is linear, then, $k \leq O(n^{1-2/q} \log^2 n)$. 
\end{theorem}
While we omit a formal argument here, as in~\cite{AlrabiahGKM23}, our methods naturally extend to give a similar lower bound for LDCs over any finite (even $n^{o(1)}$) size alphabet. We note that in a concurrent and independent work, Janzer and Manohar~\cite{JM24} prove a result that obtains a similar bound as in \Cref{main} but with a single logarithmic factor.

\subsection{Brief Overview} Let us explain our main ideas briefly (see \Cref{sec:overview} for a guided exposition with a complete proof for linear codes). After standard reductions (see \Cref{fact:normal-form}), the local decoder for a $q$-LDC can be described by a collection of $q$-uniform matchings $\cH_1, \cH_2, \ldots, \cH_k$, one for every bit $1 \leq i \leq k$ of the message. To decode the $i$-th bit in a received corrupted codeword $y$, the decoder queries a uniformly random $q$-tuple in $\cH_i$. 

\parhead{The idea from~\cite{AlrabiahGKM23,HKMMS24}:} The key technical idea in~\cite{AlrabiahGKM23} (and a combinatorial analog in~\cite{HKMMS24}) is identifying a certain \emph{well-spread} property of the hypergraph of the query sets $\bigcup_{i \in [k]} \cH_i$ that allows spectral methods based on \emph{Kikuchi Matrices} to succeed in proving that $k \lesssim n^{1-2/q}$. For any $q \geq 3$, their well-spread condition requires that no pair of $\{a,b\}$ of the codeword coordinates are queried in more than $O(\log n)$ different matchings $\cH_i$. Briefly speaking, the well-spread property is used to argue that certain  ``Kikuchi graphs'' obtained from the $\cH_i$s are approximately regular, i.e., after dropping a negligible fraction of outlier vertices, all the remaining vertices have their degrees within a $\polylog n$ factor of the average. 

The well-spread property (also called being a \emph{design} in \cite{BarakDYW11}) is a natural pseudo-random condition and is satisfied, e.g., by a collection of uniformly random hypergraph matchings with high probability. Its use in~\cite{AlrabiahGKM23} was directly motivated by works on semirandom CSP refutation and the \emph{hypergraph Moore bound}~\cite{GuruswamiKM22,HsiehKM23}, where it leads to near-optimal results. Indeed, the first step in these works is to find a simple deterministic property of random hypergraphs that suffices for the success of spectral methods for ``refuting'' (i.e., certifying unsatisfiability) random CSP instances.

Given an arbitrary $3$-LDC (i.e., arbitrary collection of matchings $\cH_i$), one must now reduce the problem to the case where the local decoding sets satisfy the well-spread condition. This is easily ensured in their setting through a greedy decomposition: we iteratively peel off pairs that occur $\gg O(\log n)$ times (call such pairs ``heavy''). After the peeling-off process ends, the remaining elements in the matching must satisfy the well-spread condition by design. If the matchings retain a constant fraction of the constraints we began with, we can apply the argument for the well-spread case referenced above. If not, we look at all the query sets peeled off. By extending the original code to have coordinates for the peeled off ``heavy pairs'', one can think of the query sets as describing an object that behaves essentially as a $2$-LDC with message length $\lesssim \log n$ (this bound corresponds to the number of times a pair was queried in the original $\bigcup_i \cH_i$). Given the known exponential lower bounds for $2$-LDCs, this is enough for a contradiction. 

The well-spread condition, when formalized for an arbitrary odd $q \geq 5$ requires control of co-degrees of all sets of size $r$ for $r \leq 5$. Indeed, such conditions were formulated as notions of regularity in  prior works~\cite{GuruswamiKM22,HsiehKM23}. Specifically, these works identified appropriate absolute thresholds $\tau_1, \tau_2, \ldots, \tau_t$ in terms of the number of vertices $n$ and related parameters, and posited that the co-degree of any size-$r$ set $Q$ (i.e., the number of hyperedges in $\cH$ containing $Q$) be upper bounded by $\tau_r$. The precise thresholds are chosen to ensure the success of a certain spectral bound on Kikuchi matrices that forms the technical core of their argument. 

For $5$-LDCs, however, such an argument fails immediately since the threshold $\tau_3$ for triples still turns out to be $\sim \log n$, and if we were to remove $\gg O(\log n)$-heavy triples, we get an object that behaves essentially as a $4$-LDC (instead of a $2$-LDC) but still with message length $\sim \log n$. Unlike $2$-LDCs, now there is no contradiction since there are $4$-LDCs~\cite{Efremenko09,Yek08} with $\gg \log^{\omega(1)} n$ message length (and moreover, the best known lower bounds for $4$-LDCs are just quadratic $k \lesssim n^{1/2}$).

\parhead{Approximate strong regularity:} Our key conceptual idea is to depart from the philosophy of reducing to well-spread $\cH=\bigcup_i \cH_i$s. Instead, we prove that a significant relaxation that we call \emph{approximate strong regularity} suffices for our spectral argument to succeed.

Approximate strong regularity does not demand any absolute upper bounds on the co-degrees of subsets of vertices in $\cH$. Instead, given a $t \leq q$, we say that a $\cH$ is $t$-approximately strongly regular if 1) all sets of size $t$ have roughly the same co-degree, say $d_t$, up to a fixed constant multiplicative slack, and 2) the co-degrees of every set of size $1 \leq r \leq q$ is appropriately bounded \emph{relative to} $d_t$. Crucially, our relative bounds do not imply any non-trivial absolute upper bounds on the co-degrees. Thus, the co-degree of triples may be $\poly(n)$ (as opposed to the $\sim \polylog n$ tolerated in prior arguments).

We show that if $\cH =\bigcup_i \cH_i$ is $t$-approximately strongly regular, then certain Kikuchi graphs built from the $\cH_i$s that arise in our spectral argument satisfy approximate regularity (i.e., dropping a few outlier vertices makes all degrees roughly equal). Our relative thresholds (see \Cref{def:goodindex}) are chosen to make a key technical piece in the argument (where we prove that certain Kikuchi graphs obtained from the hypergraph $\cH$ are approximately regular) go through (and may not look motivated in first reading).  As a result, we can argue that spectral methods succeed in establishing $k \lesssim n^{1-2/q}$ given approximate strong regularity of the underlying $\cH = \bigcup_i \cH_i$. We stress that this result holds regardless of how large $d_r$s themselves are and requires no pruning of heavy tuples, which is the key bottleneck in the prior approaches. We note that arguing the approximate regularity is the key technical component in all prior applications of the Kikuchi matrix method~\cite{GuruswamiKM22,HKMMS24,HsiehKM23,AlrabiahGKM23,KothariM23,Yankovitz24,KothariM24}.

Our second main technical component argues that any hypergraph can be decomposed (up to a negligible fraction of ``error'' hyperedges) into pieces that satisfy the required relative upper bounds on all co-degrees with respect to $d_t$ for some $t$. 

\parhead{Outlook:} Hypergraph decompositions to ensure approximate regularity of Kikuchi graphs have been a key tool in several recent results~\cite{GuruswamiKM22,HsiehKM23,AlrabiahGKM23,HKMMS24,KothariM23,KothariM24,Yankovitz24}. In all these prior applications, regularity of Kikuchi graphs was ensured by forcing absolute upper bounds on the co-degrees of all subsets in the underlying hypergraphs. The key difficulty for improved lower bounds for $q$-LDCs for odd $q \geq 5$ was precisely the inapplicability of this natural idea. We thus expect that the notion of approximate strong regularity and the fact one can make this property hold without loss of generality by hypergraph decomposition will find more applications, especially in the context of the Kikuchi matrix method. 

%% file: overview.tex
\section{Warmup: A Combinatorial Proof for the Linear Case}
\label{sec:overview}

In this section, we give a proof for the case of \emph{linear} $q$-LDCs. The argument itself is short, but we provide additional commentary to explain our key ideas.
Our proof exploits connections between LDCs and \emph{even covers} in edge-colored hypergraphs.
We will briefly discuss our approach for the non-linear case (i.e., \Cref{main}) in \Cref{sec:nonlinear-overview}.

\begin{theorem}
\label{thm:linear-LDCs}
Let $q \geq 3$ be odd and let $E:\pmo^k \rightarrow \pmo^n$ be a linear $q$-LDC that corrects up to a $\delta$ fraction of errors. Then, $k \leq O(1/\delta) \cdot  n^{1-2/q} \log^{2}n$.
\end{theorem}

In \Cref{sec:even-covers-overview,sec:kikuchi-overview}, we introduce standard notions of even covers and the Kikuchi graph.
In \Cref{sec:regular-hypergraphs}, we state our new notion of hypergraph regularity (called ``approximate strong regularity'').
This notion of regularity is a key technical idea in our work and departs from the usual \emph{well-spread} conditions on the hypergraphs appearing in prior related works (see \Cref{def:approxstrongregdef} for more details).
In \Cref{sec:right-index-overview}, we give some explanation of the conditions, and in \Cref{sec:proof-of-weak-rainbow-even-cover} we prove a key result
(\Cref{lem:even-cover-regular-hypergraph}) on the existence of ``odd-colored'' even covers when the hypergraph satisfies approximate strong regularity.
In \Cref{sec:decomposition-overview}, we discuss our decomposition algorithm that ensures approximate strong regularity for each piece except those containing a negligible fraction of the original hyperedges.
At the end of \Cref{sec:decomposition-overview}, we complete the proof of \Cref{thm:linear-LDCs}.

\subsection{LDCs and Even Covers}
\label{sec:even-covers-overview}

By standard reductions (see \Cref{fact:normal-form}), we can assume that the $q$-LDC is in the normal form. Thus, there are $q$-uniform matchings of size $\Omega(n)$, say $\cH_1, \cH_2, \ldots, \cH_k$, that correspond to the local decoding sets for each of the message bits $1 \leq i \leq k$. Further, the decoder itself is linear. Thus, if $x= E(b)$ for a $q$-LDC $E: \pmo^k \to \pmo^n$, then for every $C \in \cH_i$, it must hold that $\prod_{u \in C} x_u = b_i$. 

An \emph{even cover} in a hypergraph is a collection of hyperedges that use every vertex an even number of times ---  a natural linear algebraic generalization of cycles in graphs. When viewing a hyperedge as the coefficient vector of a linear form over $\F_2$, an even cover is simply a subset of equations that add up to $0$ in $\F_2$. 

We begin with a simple connection between linear LDCs and even covers in edge-colored hypergraphs that has been standard in recent works on LDCs~\cite{AlrabiahGKM23,HKMMS24}.

\begin{fact}[Lemma 2.7 of \cite{HKMMS24}] \label{fact:even-cover}
    Let $E: \on^k \to \on^n$ be a binary, linear $q$-query LDC associated with $q$-uniform matchings $\cH_1, \cH_2, \dots, \cH_k$.
    Suppose we color each hyperedge in $\cH_i$ with color $i$.
    Then, any even cover in $\cH = \bigcup_{i\in [k]} \cH_i$ must use color $i$ an even number of times for each $i\in [k]$.
\end{fact}
This fact itself is simple to prove. Suppose an even cover uses a subset $S_i \subseteq \cH_i$ of hyperedges from the $i$-th matching, and suppose $|S_i|$ is odd. Then, for any message vector $b \in \pmo^k$ and the corresponding codeword $x = E(b)$, it must hold that $x_C = \prod_{u \in C}x_u= b_i$ for every $C \in S_i$. If we take product of the equations corresponding to edges in the even cover, then the LHS must square out to $1$ while the RHS is the product of $b_i$s for $i$ such that $S_i$ has odd size. This is a contradiction if we set such message bit to $-1$.

\parhead{LDC Lower Bounds and Weakly Rainbow Even Covers:} Thus, to prove that $k \ll n^{1-2/q}$ in any linear $q$-LDC, it suffices to show that if $k \gg n^{1-2/q}$ (i.e., $\cH = \bigcup_{i} \cH_i$ is a sufficiently dense hypergraph) then there is an even cover in $\cH$ that uses some color an \emph{odd} number of times.

This goal has a nice combinatorial interpretation. Notice that $\cH$ is a \emph{properly hyperedge-colored} hypergraph ---- since the hyperedges incident on any single vertex are all of distinct colors as the $\cH_i$s are matchings. Thus, our goal is to show that a properly edge-colored hypergraph that is dense enough must have an even cover that uses some color an odd number of times. In fact, we will prove the stronger statement that such a hypergraph must contain an even cover that uses exactly one hyperedge of some color. Let us make this into a definition before we move on:

\begin{definition}[Weak Rainbow Even Covers]
A \emph{weakly rainbow} even cover in a properly hyperedge-colored hypergraph $\cH$ is an even cover that uses exactly one hyperedge of some color. 
\end{definition}

Showing the existence of weakly rainbow even covers in dense enough hypergraphs is a variant of the well-studied rainbow cycle problem in graphs. Such problems were introduced in the pioneering work of~\cite{KMSV07}. Their work included, among many results, a \emph{weak rainbow cycle} theorem for graphs: every properly edge-colored graph on $n$ vertices with average degree $\geq O(\log n)$ has a cycle that uses some color exactly once\footnote{We call this version ``weak'' to distinguish from the better-known rainbow cycle conjecture asks for a cycle where every color is used at most once.}. There is a long line of work building on their work
including the ones leading up to the recent (almost) resolution of the related \emph{rainbow cycle conjecture}~\cite{DAS2013905,Janzer2020RainbowTN,Tomon2024Robust,Kim2024, JanzerSudakov,ABSZZ}. Here, let us state a generalization that we will apply as a black-box:
\begin{fact}[Lemma 2.5 of \cite{HKMMS24}] \label{lem:weak-rainbow-cycle}
    Let $G$ be an $n$-vertex graph where each edge in $G$ is assigned a set of $s$ colors, and suppose $G$ has average degree $d \geq 40s\log n$. Suppose that for every vertex $v \in G$ and color $c$, the number of edges incident on $v$ whose assigned set of colors contains $c$ is at most $d/20s \log n$. Then, $G$ contains a closed walk, of size at most $2\log n$, such that some color appears exactly once.
\end{fact}

Our main result for linear LDCs is a natural hypergraph analog of the weak rainbow cycle theorem for graphs (see \Cref{lem:even-cover-regular-hypergraph} for the key technical component). We will prove this by a reduction to edge-colored graphs and applying \Cref{lem:weak-rainbow-cycle}. This reduction is based on Kikuchi graphs that we next introduce.

\subsection{The Kikuchi Graph}
\label{sec:kikuchi-overview}

As in the prior works, a key technical idea is the use of Kikuchi graphs. 
In the even $q$ case, for $\ell \leq n$, the level-$\ell$ Kikuchi graph associated with an edge-colored hypergraph $\cH$ is an edge-colored graph defined as follows:
\begin{itemize}
    \item The vertex set is $\binom{[n]}{\ell}$, i.e., all $\ell$-sized subsets of $[n]$.
    \item For each hyperedge $C \in \cH$ with color $i$, we add edges $(S,T)$, denoted by $S \xleftrightarrow{C} T$, such that $S \oplus T = C$ and $|S \cap C| = |T \cap C| = q/2$, and color them with $i$.
\end{itemize}
The key to the utility of these graphs is an elementary but important connection~\cite{GuruswamiKM22} between cycles/walks in the Kikuchi graph and even covers in $\cH$. As a result, finding a weak rainbow even cover in $\cH$ reduces to finding a weak rainbow cycle in the edge-colored Kikuchi graph. For even $q$, this approach easily recovers the best known lower bounds for $q$-LDCs~\cite{GKST06,KdW04}.

For odd $q$, the challenge with this approach is that the Kikuchi graph above does not make sense (for $S \oplus T = C$ and $|S| = |T| = \ell$, $C$ must be of even size). To get around this issue, all prior works~\cite{GuruswamiKM22,HsiehKM23,AlrabiahGKM23,HKMMS24} use variants of Kikuchi graphs where each transition corresponds to a pair $(C,C')$ of hyperedges (instead of a single $C$ for the even case above).

Let us informally describe the Kikuchi graph we use (see \Cref{def:kikuchi} and the illustration \Cref{fig:KikuchiGraph}). For a parameter $t \in \{1,2,\dots,q-1\}$ (that will eventually be chosen carefully), our Kikuchi graph has transitions corresponding to pairs $(C,C')$ that intersect in $t$ vertices. In this case, note that $C \oplus C'$ is a set of size $2(q-t)$, which is even.
Our Kikuchi graph takes all such pairs in $\cH$ and includes edges $S \xleftrightarrow{C,C'} T$.\footnote{For the arguments that follow, we need to setup this graph so that $S$ and $T$ each intersect both $C,C'$ in equal, up to a slack of at most $1$ element, but we omit a discussion of this issue here.}

Since $\cH$ is properly hyperedge-colored, $C$ and $C'$ must have different colors $i \neq j$. Thus, each edge of the Kikuchi graph is associated with $2$ distinct colors.  We note that in \cite{AlrabiahGKM23}, the $t$ above is hardcoded to $1$; in our case we must carefully choose $t$ as a function of the hypergraph.

\parhead{Weak rainbow even covers by applying \Cref{lem:weak-rainbow-cycle} to Kikuchi graphs:}
Our Kikuchi graph naturally has edges that have a pair of colors. We say that a graph (where each edge can have multiple colors) is $\Delta$-properly edge-colored if, for each color $i$, every vertex is incident to at most $\Delta$ edges whose color set includes the color $i$.

We want to apply \Cref{lem:weak-rainbow-cycle} to the Kikuchi graph with $s=2$. For this, we need to argue that the Kikuchi graph contains a large subgraph that
\begin{enumerate}[(1)]
    \item has average degree $\ol{d}(K) \geq O(\log N)$,
    \item is $\ol{d}(K)/O(\log N)$-properly edge-colored --- that is, given any color $i$, the number of edges with a pair of colors including $i$ incident on any vertex is at most $\ol{d}(K)/O(\log N)$. 
\end{enumerate}

The average degree of the Kikuchi graph is easily described as a function of $\ell$ and $|\cH|$ (see \Cref{fact:bound-D}), and setting $\ell=n^{1-2/q}$ will guarantee $\ol{d}(K) \geq O(\log N)$. The sticking point is ensuring that the Kikuchi graph is $\ol{d}(K)/O(\log N)$-properly edge-colored.

\begin{remark}[Prior works] \label{rem:prior-works-overview}
    Ensuring the proper edge-coloring property (which is essentially equivalent to approximate regularity of the Kikuchi graphs) is in fact the key technical step in all applications of the Kikuchi matrix method so far~\cite{GuruswamiKM22,HsiehKM23,HKMMS24,AlrabiahGKM23,KothariM23,KothariM24}. In all these prior works, such a property was ensured when the underlying hypergraph satisfies a natural well-spread condition. This well-spread condition in the work on $3$-LDCs~\cite{AlrabiahGKM23} asks that in $\cH$, no pair of vertices appear in more than $O(\log n)$ hyperedges. The remaining part of the argument is then giving a different argument when $\cH$ does not satisfy this well-spread property. In all prior applications, this is done by constructing a hypergraph of lower uniformity. For the case of $3$-LDCs, this gives a graph and the $O(\log n)$ threshold for pairs being ``heavy'' above translates into its average degree being $\gg O(\log n)$. This average degree is sufficient to get cycles in such a graph that use some color exactly once. For the case of $5$-LDCs, however, the same reduction produces a hypergraph of uniformity $4$ with average degree $O(\log n)$. Such a hypergraph can manifestly only have even covers that use every color even number of times (for e.g., those coming from the Efremenko-Yekhanin matching vector codes~\cite{Yek08,Efremenko09}!).
\end{remark}

Our main idea to overcome this obstacle is finding a significantly more general condition --- \emph{approximate strong regularity} --- on hypergraphs such that the associated Kikuchi graphs still possess the $\Delta$-properly edge-colored property. Our more general condition can violate the well-spread property  of hypergraphs drastically, and pairs, triples, quadruples, etc.\ can appear in an arbitrarily large number of hyperedges.

\subsection{Approximate Strong Regularity}
\label{sec:regular-hypergraphs}
Given a hypergraph $\cH$, the co-degree of a set $Q$ of vertices, which we denote as $\codegree{\cH}{Q}$, is the number of hyperedges in $\cH$ that contain $Q$. Moreover, we denote $d_{\cH,t} \coloneqq \max_{|Q| =t} \codegree{\cH}{Q}$ for all $t\in [q]$.
We will omit the dependence on $\cH$ and write $d_t$ for simplicity.
Note that we have $d_1 \geq d_2 \geq \cdots \geq d_q$ and $d_1 \geq |\cH|/n$.

Our key technical piece shows the existence of a weakly rainbow even cover if $\cH$ satisfies a property called approximate strong regularity. This property demands that there be some $t \leq q$ such that 1) $\cH$ be partitioned into groups of roughly the same size $d_t$ such that hyperedges in each group all intersect in a fixed set of size $t$, and, 2) the co-degrees of any subset of $r$ vertices (for $1 \leq r\leq q$ is small \emph{relative to} $d_t$. This demand for relative as opposed to absolute upper bounds is a key departure from similar notions of regularity appearing in prior works (and their insufficiency is a key reason for their inapplicability in proving improved lower bounds for odd $q \geq 5$). 

The precise relative bounds below can look daunting. The choice is dictated by a key technical piece in our proof, where we show the approximate regularity of Kikuchi graphs that arise in our analysis. We note that arguing the approximate regularity of Kikuchi graphs is the centerpiece in all applications of the Kikuchi matrix method so far~\cite{GuruswamiKM22,HsiehKM23,HKMMS24,AlrabiahGKM23,KothariM23,Yankovitz24,KothariM24}). We invite the reader to ignore the quantitative requirements in the first reading of this overview --- we explain their origin in the next subsection (\Cref{sec:right-index-overview}).

\begin{definition}[Good index $t$]
\torestate{
\label{def:goodindex}
Let $d_1\geq d_2 \geq \cdots \geq d_q$. We say that an index $t$ is \emph{good} with respect to the tuple $(d_1, d_2, \ldots, d_q)$ if the following conditions hold:
\begin{enumerate}[(1)]
\item $d_r/d_{t}\leq n^{1-\frac{2r}{q}}$ for every $1 \leq r \leq \lceil \frac{q-t}{2} \rceil$, 
\item $d_r/d_{t} \leq n^{-\frac{2}{q}(r-t) + \frac{1}{q}(t-\1(t\text{ even}))}$ for every $t \leq r \leq \lfloor \frac{q+t}{2} \rfloor$,
\item If $t <q/2$, then, $d_{t} \geq d_1 n^{-\frac{2}{q}(t-1)}$, and, if $t>q/2$, then, $d_{t} \geq d_1 n^{-1+2/q}$. \label{item:d-t-lower-bound}
\end{enumerate}
}
\end{definition}

We now define $t$-approximate strong regularity of a $q$-uniform hypergraph $\cH$ that posits the above conditions are satisfied for some $1 \leq t \leq q$.  

\begin{definition}[Approximate Strong Regularity]
\torestate{
\label{def:approxstrongregdef}
For a given $t\in[q]$, a $q$-uniform hypergraph $\cH$ on $[n]$ is called $t$-approximately strongly regular if there exists a partitioning $\cH = \cH^{(t)}_1 \sqcup \cdots \sqcup \cH^{(t)}_{p_t}$ such that:
    \begin{enumerate}[(1)]
        \item For every $\theta\in[p_t]$, there is some $Q_\theta \in \binom{[n]}{t}$ such that all $C \in \cH^{(t)}_\theta$ contain $Q_\theta$.
        \item $d_{t}/2\leq \big|\cH^{(t)}_\theta\big|\leq d_{t}$ for all $\theta\in[p_t]$.
        \item $t$ is a good index 
        with respect to the sequence $d_1 \geq \cdots \geq d_q$ (recall that $d_r:= \max_{|Q| = r} d_{\cH,Q}$). \label{item:good-index}
    \end{enumerate}
}
\end{definition}

Let us interpret this definition.
First, the numbers $d_r$ give an upper bound on the co-degree of every size-$r$ subset of vertices in $\cH$.
Thus, the first $2$ conditions imply that the hypergraph can be partitioned into pieces, each containing a $(d_t/2)$-heavy subset of $t$ vertices.
Condition (\ref{item:good-index}) states that all other co-degrees are bounded relative to $d_t$, and moreover, item (\ref{item:d-t-lower-bound}) of \Cref{def:goodindex} gives a lower bound on $d_t$ in terms of $d_1$, which we know is at least $|\cH|/n$.

We can now state our weak hypergraph rainbow lemma for approximately strongly regular hypergraphs. 
\begin{lemma}[Weak Rainbow Bound for Approximately Strongly Regular Hypergraphs] \label{lem:even-cover-regular-hypergraph}
    Let $q\in \N$ be odd and $t\in [q]$.
    There is a universal constant $A$ depending only on $q$ such that the following holds:
    Let $\cH$ be a properly hyperedge-colored, $q$-uniform, $t$-approximately strongly regular hypergraph on $[n]$ with average degree $k \geq A n^{1-\frac{2}{q}} \log n$.

    Then, $\cH$ contains a weak rainbow even cover. 
\end{lemma}

Before we prove \Cref{lem:even-cover-regular-hypergraph}, let us first explain the conditions of a good index in \Cref{def:goodindex}.
One can also see these conditions arise directly in the proof of \Cref{lem:even-cover-regular-hypergraph} in \Cref{sec:proof-of-weak-rainbow-even-cover}.
Later in \Cref{sec:decomposition-overview}, we will show that approximate strong regularity can be assumed without loss of generality by decomposing the given hypergraph (though losing a log factor).

\subsection{Avoiding clustering of edges in the Kikuchi graph}

\label{sec:right-index-overview}

In this section, we provide a high-level explanation (with a concrete example) showing why the conditions for a good index in \Cref{def:goodindex} imply that the Kikuchi graph satisfies the proper coloring condition required in \Cref{lem:weak-rainbow-cycle}.

The Kikuchi graph $K$ at level $\ell$ we described for $\cH$ is obtained by picking an integer $t \leq q$ and then constructing edges corresponding to pairs $(C,C')$ such that $C,C'$ intersect in a set of size $t$. The choice of $t$ needs to be carefully done to ensure the $\ol{d}(K)/O(\log N)$-proper coloring property of the resulting Kikuchi graph.
Recall that to apply \Cref{lem:weak-rainbow-cycle}, we need 1) the average degree of the Kikuchi graph to be $\ol{d}(K)\geq O(\log N)$, and 2) the graph is $\sim \ol{d}(K)/\log N$-properly edge-colored.
Let's set $t=1$ for the sake of illustration and first compute the average degree $\ol{d}(K)$ of the Kikuchi graph. Since $d_1 \sim k \sim n^{1-2/q}$ (each matching is near perfect), the number of pairs $(C,C')$ in $\cH$ that intersect in one vertex is $\sim nk^2$. Each such pair then has $2(q-1)$ non-overlapping vertices. A random vertex $S$ of the Kikuchi graph behaves essentially like a $\ell/n$-biased random set, and thus, a given $(C,C')$ generates an edge incident on it with probability $(\ell/n)^{q-1}$. Thus, $\ol{d}(K) \sim (\ell/n)^{q-1} nk^2 \geq \log N \sim \ell \log n$ if $\ell = n^{1-2/q}$ and $k \geq \tilde{O}(n^{1-2/q})$.

Now fix a color $i$. We must argue that the number of pairs $(C,C')$ that induce an edge on a vertex $S$ of the Kikuchi graph $K$ is\footnote{We note that when $t>1$, the degree of the Kikuchi graph will in general be $\gg O(\log N)$ and depend on $d_t$. In that case, the number of edges incident on a typical $S$ with a fixed color $i$ that we can tolerate will scale relative to $d_t$. } at most $\ol{d}(K)/O(\log N) \sim \polylog n$. It suffices to argue that this holds for almost all vertices since we can then delete a negligible fraction of edges to ensure it for all without noticeably hurting the average degree. 

We next show a concrete scenario where the above property fails  (here, due to the co-degree $d_3$ of triples being too large). Let us set $q=5$ for this illustration since the issue already pops up for this first interesting case. For every $C \in \cH_i$, we can have $\sim k$ different $C'_1, C'_2, \ldots,$ hyperedges in $\cH$ that intersect $C$ in exactly one vertex, say $u$. All such pairs $(C,C'_j)$ contribute edges in our Kikuchi graph with a color pair containing $i$.  Suppose that there are $d_3$ different $C'_j$s all containing a fixed pair $\{v,w\}$. In this case, $d_3$ different $C'_j$s contain the same triple $\{u,v,w\}$.
See \Cref{fig:clustering_diagram} for an illustration.

\begin{figure}[ht]
    \centering
    \includegraphics[width=0.85\linewidth]{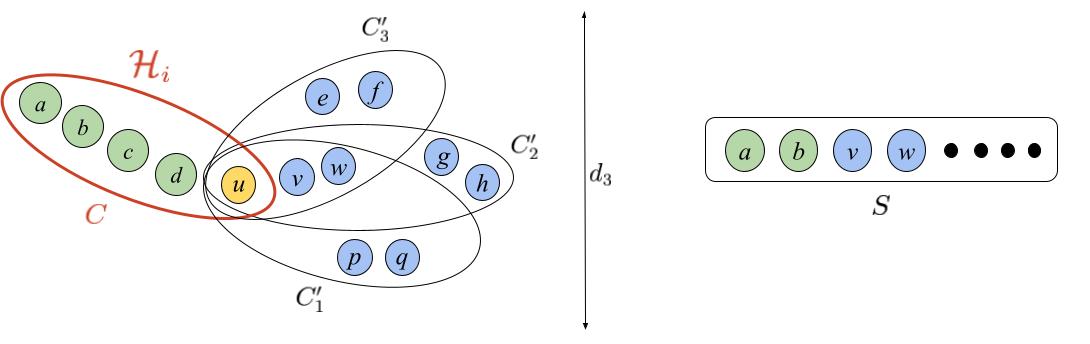}
    \caption{This example shows, for $q = 5$, a scenario where sets $S$ as shown have $\geq d_3$ edges where one of the colors is $i$.}
    \label{fig:clustering_diagram}
\end{figure}

A pair $(C,C')$ contributes an edge on $S$ if $S$ intersects $C$ and $C'$ in $(q-1)/2=2$ elements each for $q=5$. $S$ already contains $2$ vertices from $C$ and by containing the pair $\{v,w\}$ two vertices from \emph{every one of the $d_3$ different} $C'_j$s. Thus, $S$ has an edge corresponding to \emph{every one of the} $(C,C'_j)$ where $C'_j$ contains $\{v,w\}$ and thus  at least $d_3$ edges with color $i$. If $d_3 \gg \polylog n$ (it can be $n^{\Omega(1)}$ in general!), we have failed to satisfy the proper coloring requirement.

One can generalize the above bad example into a setting where setting $t=1$ in our Kikuchi graph simply fails to satisfy the proper coloring condition required in~\Cref{lem:weak-rainbow-cycle} for a constant fraction of vertices. More generally, for any given $t$, the above analysis necessitates that the ratio $d_r/d_t$ be appropriately upper bounded as a function of $r$ and $t$.
Such restrictions give rise to the specific conditions in \Cref{def:goodindex}, where each condition corresponds to a bad scenario analogous to the example above. It turns out that we can argue (done in full in the next subsection!) that the list of bad scenarios the above conditions rule out is exhaustive. In other words, if every one of the conditions above is met, then, we can establish a $\ol{d}/O(\log N)$ upper bound on the number of edges of a fixed color $i$ (out of a pair of colors) incident on almost all vertices of the Kikuchi graph.

\subsection{Proof of \texorpdfstring{\Cref{lem:even-cover-regular-hypergraph}}{Lemma~\ref{lem:even-cover-regular-hypergraph}}}
\label{sec:proof-of-weak-rainbow-even-cover}

Let us now give a full proof of \Cref{lem:even-cover-regular-hypergraph}.
\begin{proof}[Proof of \Cref{lem:even-cover-regular-hypergraph}]
    Recall that $d_1 \geq d_2 \geq \cdots \geq d_q$ denote the co-degrees of $\cH$, and $t$ is a good index as in \Cref{def:goodindex}.
    We first claim that
    \begin{align*}
        d_t \geq \max\braces*{n^{1-\frac{2t}{q}},\ 1} \cdot \Omega(\log n) \mper
        \numberthis \label{eq:dt-lower-bound}
    \end{align*}
    To see this, note that by item~(\ref{item:d-t-lower-bound}) of \Cref{def:goodindex}, we have $d_t \geq d_1 \cdot \max\{n^{-\frac{2}{q}(t-1)},\ n^{-1+\frac{2}{q}}\}$.
    Moreover, $d_1$ is at least the average degree, which is assumed to be $\geq \Omega(n^{1-2/q} \log n)$.
    This establishes \Cref{eq:dt-lower-bound}.

    Next, let $\ell \coloneqq n^{1-\frac{2}{q}}$,
    and consider the level-$\ell$ Kikuchi graph $K$ as described in \Cref{sec:kikuchi-overview} by canceling $t$-tuples.
    By $t$-approximate strong regularity (\Cref{def:approxstrongregdef}), $\cH$ has a partition $\{\cH_{\theta}^{(t)}\}_{\theta\in[p_t]}$, each of size at least $d_t/2$.
    Thus, we can form $\sum_{\theta=1}^{p_t} \binom{|\cH_{\theta}^{(t)}|}{2} \geq \Omega(d_t) \sum_{\theta=1}^{p_t} \big|\cH_{\theta}^{(t)}\big| \geq \Omega(knd_t)$ pairs of hyperedges.
    Thus, the average degree of the Kikuchi graph $K$ is
    \begin{align*}
        \ol{d}(K) \geq \Omega\parens*{\frac{\ell}{n}}^{q-t} \cdot kn d_t
        \geq \Omega(d_t) \cdot n^{\frac{2}{q}(t-1)} \log n \mper
    \end{align*}
    In particular, by \Cref{eq:dt-lower-bound}, we have $\ol{d}(K) \geq \Omega(n^{1-\frac{2}{q}} \log^2 n) \geq \omega(1) \cdot \log |V(K)|$, since $|V(K)| = \binom{2n}{\ell} \leq (2n)^{\ell}$.
    This establishes the average degree lower bound required in \Cref{lem:weak-rainbow-cycle}.

    Next, define
    \begin{align*}
        \Delta \coloneqq \frac{\ol{d}(K)}{100\ell \log n} \geq A' d_t \cdot n^{\frac{2t}{q}-1} \mcom
        \numberthis \label{eq:Delta}
    \end{align*}
    where $A'$ is some large constant depending on $q$.
    It suffices to show that there is a subgraph of $K$ which is $\Delta$-properly edge-colored and has average degree $\ol{d}(K)/2$.
    This would complete the proof since \Cref{lem:weak-rainbow-cycle} guarantees the existence of a cycle, of length at most $2\log |V(K)| = O(\ell \log n)$, that uses a color exactly once.

    Fix any $C_1 \neq C_2 \in \cH_{\theta}^{(t)}$ of colors $i,j$ respectively.
    We know that $i \neq j$ since $\cH$ is properly edge-colored and $|C_1 \cap C_2| \geq t > 0$.
    Recall that the pair $(C_1,C_2)$ forms a matching in $K$, denoted by $\{(S,T): S \xleftrightarrow{C_1,C_2} T\}$.
    We will remove an edge $(S,T)$ in this matching if $S$ or $T$ is incident to more than $\Delta$ edges containing color $i$ or $j$.
    Our goal is to prove that this process removes at most $1/2$ fraction of the edges.

    Consider a random edge $S \xleftrightarrow{C_1,C_2} T$ among the matching edges, which can be sampled by first choosing the symmetric difference $S \oplus T$ from $C_1 \setminus Q_{\theta}$ and $C_2 \setminus Q_{\theta}$, and then sampling $\ell - \floor{\frac{q-t}{2}}$ vertices for $S \cap T$.
    By Markov's inequality,
    \begin{align*}
        \Pr_{(S,T)} \bracks*{(S,T) \text{ is removed}}
        \leq \frac{1}{\Delta} \cdot \E_{(S,T)}\bracks*{\text{number of edges incident to $(S,T)$ with color $i$ or $j$}} \mper
    \end{align*}
    We will upper bound the right-hand side via the co-degree guarantees from \Cref{def:goodindex}.
    There are two possible ways that an incident edge can have color $i$ (same analysis for $j$).
    \begin{itemize}
        \item $(S,T)$ is incident to an edge formed by $C_1, C_3$.
        In this case, $C_3$ must also be in $\cH_{\theta}^{(t)}$ and thus must have $|C_2 \cap C_3| \geq |Q_{\theta}| = t$.
        Suppose $|C_2 \cap C_3| = t+a$ for some $a \geq 0$.
        For a random $(S,T)$ to be incident to an edge formed by $(C_1,C_3)$, one must sample at least $\max\{\floor{\frac{q-t}{2}}-a, 0\}$ elements from $C_3 \setminus C_2$ to complete $S \cap T$.
        This probability is at most $O\left(\frac{\ell}{n}\right)^{\max\{ \floor{\frac{q-t}{2}}-a, 0\}}$.
        On the other hand, there are at most $O(d_{t+a})$ such $C_3$ with $|C_2 \cap C_3| = t+a$.
        Thus, using \Cref{eq:Delta}, the contribution from this case is at most
        \begin{align*}
            \frac{1}{\Delta} \sum_{a=0}^{q-t} O(d_{t+a}) \cdot O\parens*{\frac{\ell}{n}}^{\max\{ \floor{\frac{q-t}{2}}-a, 0\}}
            &\leq \frac{O(1)}{A'} \sum_{a=0}^{\floor{\frac{q-t}{2}}} \frac{d_{t+a}}{d_t} \cdot n^{1-\frac{2t}{q}} \cdot n^{-\frac{2}{q} (\floor{\frac{q-t}{2}}-a)} \\
            &\leq \frac{O(1)}{A'} \sum_{a=0}^{\floor{\frac{q-t}{2}}} \frac{d_{t+a}}{d_t} \cdot n^{\frac{2a}{q} - \frac{1}{q}(t - \1(t\text{ even}))} \\
            &\leq O(1/A') \mper
        \end{align*}
        Here, we use that $d_s$ is decreasing with $s$ to handle the terms with $a > \floor{\frac{q-t}{2}}$, and the assumption $d_{t+a} \leq d_t \cdot n^{-\frac{2a}{q} + \frac{1}{q}(t-\1(t \text{ even}))}$.

        \item $(S,T)$ is incident to an edge formed by $(C_3, C_4)$, where $C_3 \neq C_1$ and $C_3$ has color $i$.
        Since $\cH$ is properly edge-colored, we have $C_1 \cap C_3 = \varnothing$ and $C_3, C_4 \in \cH_{\theta'}^{(t)}$ for some $\theta' \neq \theta$.
        Suppose $|C_4 \cap C_2| = s$.
        Then, to form $S\cap T$, one must sample $\floor{\frac{q-t}{2}}$ (resp.\ $\ceil{\frac{q-t}{2}}$) elements from $C_3$ and $\max\{\ceil{\frac{q-t}{2}}-s, 0\}$ (resp.\ $\max\{\floor{\frac{q-t}{2}}-s, 0\}$) elements from $C_4 \setminus C_2$, hence at least $\max\{q-t-s, \floor{\frac{q-t}{2}}\}$.
        The number of such pairs $(C_3, C_4)$ can be upper bounded as follows: (1) if $s=0$, then there are $n$ choices for $C_3$ with color $i$, which belongs to some $\cH_{\theta'}^{(t)}$, and $|\cH_{\theta'}^{(t)}| \leq d_t$ choices for $C_4$,
        (2) if $s \geq 1$, then there are $O(d_s)$ choices for $C_4$, which belongs to some $\cH_{\theta'}^{(t)}$, and at most $1$ choice for $C_3 \in \cH_{\theta'}^{(t)}$ of color $i$.
        Thus, using \Cref{eq:Delta}, the contribution from this case is at most
        \begin{align*}
            & \frac{1}{\Delta} \parens*{nd_t \cdot O\parens*{\frac{\ell}{n}}^{q-t} + \sum_{s=1}^{q} O(d_{s}) \cdot O\parens*{\frac{\ell}{n}}^{\max\{q-t-s, \floor{\frac{q-t}{2}}\}} } \\
            & \leq \frac{O(1)}{A'} \cdot n^{1-\frac{2t}{q}}   \parens*{n \cdot n^{-\frac{2}{q}(q-t)} + \sum_{s=1}^{\ceil{\frac{q-t}{2}}} \frac{d_s}{d_t} \cdot n^{-\frac{2}{q}(q-t-s)} } \\
            &\leq O(1/A')  \mper
        \end{align*}
        Again, we use that $d_s$ is decreasing with $s$ to handle the terms with $s > \ceil{\frac{q-t}{2}}$, and the assumption $d_s \leq d_t \cdot n^{1-\frac{2s}{q}}$.
    \end{itemize}
    Therefore, for large enough $A'$ (depending only on $q$), we have that $\Pr_{(S,T)}[(S,T) \text{ is removed}] < 1/2$.
    Thus, the edge deletion process will remove at most $1/2$ of the edges.
    This completes the proof.
\end{proof}

\subsection{Regularity Decomposition}
\label{sec:decomposition-overview}

Our notion of approximate strong regularity (\Cref{def:approxstrongregdef}) requires a good index $t$ as in \Cref{def:goodindex}.
For all we know, such an index may not even exist.

\parhead{There is always a good $t$:}
One of our main technical insights is that given any hypergraph, \emph{for every possible sequence of co-degrees} $d_1 \geq d_2 \geq \cdots \geq d_q$, there is always an index $t \in [q]$ that satisfies all the conditions in \Cref{def:goodindex}!
The proof is elementary and short. We argue that either a $t_0<q/2$ that extremizes a natural measure is good, or, if not, then an extremal $d_t$ for $t>q/2$ that violates the requirements for $t_0$ being a good index works. See \Cref{lem:good-index-exists} for the formal statement and \Cref{sec:numberslemma} for the proof.
We remark that such a proof is possible precisely because all inequalities in \Cref{def:goodindex} are relative.

\medskip
\noindent
Given that a good index always exists, we next describe how to decompose the hypergraph into $\sim \log n$ pieces, where each piece is $t$-approximately strongly regular for some $t\in [q]$.
We refer readers to \Cref{decompsec} for more details.

\parhead{Decomposition in prior works:}
All recent works on even covers and LDC lower bounds~\cite{GuruswamiKM22,HsiehKM23, AlrabiahGKM23,HKMMS24} decompose the hypergraph into $\cH^{(1)},\dots,\cH^{(q)}$ as follows: Carefully choose thresholds $\tau_s$ for all $s\in [q]$.
Then, starting with $t=q$, (1) if there exists a $Q \in \binom{[n]}{t}$ such that $|\{C \in \cH: Q \subseteq C\}| \geq \tau_t$, then add these hyperedges to $\cH^{(t)}$ and remove them from $\cH$,
(2) if no such $Q$ exists, then decrement $t$ and repeat unless $t=1$.

By this decomposition, we have the guarantee that for each $\cH^{(t)}$, the co-degrees satisfy $d_s < \tau_s$ for all $s > t$, otherwise these hyperedges would have been put in $\cH^{(s)}$.
Unfortunately, this does not suffice for us, as \Cref{def:goodindex} requires bounds on co-degrees for all $s$.

\parhead{Approximate Strong Regularity decomposition:}
Rather than decomposing to satisfy some fixed thresholds on the co-degrees, we do the following:
\begin{enumerate}[(1)]
    \item Set the ``current'' hypergraph to be $\cH$.
    \item Apply \cref{lem:good-index-exists} on its co-degree sequence $d_1 \geq \cdots \geq d_q$ to obtain a good index $t \in [q]$.
    \item Greedily extract hyperedges containing ``heavy'' $t$-tuples from the current hypergraph, i.e. take $Q \in \binom{[n]}{t}$ that maximizes $\big|\cH^{\mathrm{curr}}_{|Q}\big| = |\{C \in \cH': Q \subseteq C\}|$, and remove $\cH^{\mathrm{curr}}_{|Q}$ from $\cH^{\mathrm{curr}}$.
    \item Repeat step (3) until the $d_t$-value in $\cH^{\mathrm{curr}}$ falls below half of its original value (in the parent hypergraph).
    \item At this point, all the $\cH^{\mathrm{curr}}_{|Q}$s which were removed form a $t$-approximately strongly regular hypergraph. We designate the hypergraph remaining as the next ``current'' hypergraph, and repeat the above process.
\end{enumerate}

Note that in the above process, the value of $d_t$ in the ``current'' graph goes down by a factor of $2$ for some $t\in[q]$ for every iteration, and thus the above process can only run for $\sim q \log_2|\cH|$ steps (if the process goes on for longer, then some $d_t$ gets halved $\geq  q \log_2|\cH|/q \geq \log_2 |\cH|$ times).

Therefore, this decomposition produces $O(q\log n)$ subhypergraphs, each of which is approximately strongly regular, and one of them contains at least $\Omega(\frac{1}{q\log n})$ fraction of $\cH$.

\begin{claim} \label{lem:decomp-overview}
    There exists a subhypergraph $\cH^{\pi} \subseteq \cH$ with at least $\Omega\left(\frac{|\cH|}{q\log n}\right)$ hyperedges and is $t$-approximately strongly regular for some $t\in [q]$.
\end{claim}

See \Cref{lem:new_bucketing_arpon_version} and its proof in \Cref{decompsec} for a more general result that we need for the non-linear case.
For linear LDCs, the largest piece in the decomposition suffices for our analysis.

Now, \Cref{thm:linear-LDCs} is a simple corollary of \Cref{lem:even-cover-regular-hypergraph,lem:decomp-overview}.
\begin{proof}[Proof of \Cref{thm:linear-LDCs}]
    Towards a contradiction, let the message length satisfy $k \geq \Omega\big(\delta^{-1} n^{1-\frac{2}{q}} \log^{2} n\big)$.
    In the underlying hypergraph $\cH$, each of the $k$ matchings has size $\Omega(\delta n)$.
    Let $\cH^{\pi}$ be the subhypergraph of $\cH$ from \Cref{lem:decomp-overview}.
    Then, $\cH^{\pi}$ has average degree $\Omega\big(n^{1-\frac{2}{q}} \log n\big)$, and is $t$-approximately strongly regular for some $t\in[q]$, thus we can apply \Cref{lem:even-cover-regular-hypergraph} and conclude that $\cH^{\pi}$ (hence also $\cH$) must contain a weak rainbow even cover.
    This is a contradiction via \Cref{fact:even-cover}.
\end{proof}

\subsection{Non-linear LDC Lower Bounds}
\label{sec:nonlinear-overview}

To prove lower bounds for non-linear $q$-LDCs, we depart from the combinatorial lens, since merely showing the existence of a weak rainbow even cover does not suffice.
Instead, we use the connection to CSP refutation discovered by \cite{AlrabiahGKM23}:
Given a $q$-query normally decodable code (\Cref{def:normalldc}) $E: \pmo^k \to \pmo^n$ associated with $q$-uniform matchings $\cH_1,\cH_2,\dots,\cH_k$, for each message $b\in \pmo^k$, we can construct a $q$-XOR instance $\Psi_b$ over $x\in \pmo^n$ with clauses $x_C = b_i$ for each $i\in [k]$ and $C \in \cH_i$.
Here, we denote $x_C = \prod_{u\in C} x_u$.

\Cref{def:normalldc} can be interpreted as that every clause $C$ in $\Psi_b$ is satisfied by $x = E(b)$ for ``most'' messages $b \in \pmo^k$.
This implies that the value of $\Psi_b$ --- the maximum number of satisfied constraints among $x$ --- must be large for ``most'' $b\in \pmo^k$.
Thus, to rule out the existence of a $q$-LDC (with $k$ too large), it suffices to prove that for a random $b \in \pmo^k$, $\Psi_b$ is \emph{highly unsatisfiable} --- no assignment satisfies $\frac{1}{2} + \Omega(1)$ fraction of the clauses.
This is closely related to semirandom CSP refutation since both problems involve proving unsatisfiability of instances where the underlying hypergraph structure is arbitrary but the signs are random.

We write $\Psi_b$ as a degree-$q$ polynomial in $x = (x_1,\dots,x_n)$: $\Psi_b(x) = \sum_{i=1}^k \sum_{C \in \cH_i} b_i x_C$.
Denoting $m = |\cH| = \sum_{i=1}^k |\cH_i|$, we have $\Psi_b(x) \in [-m,m]$ for all $x\in \pmo^n$.
Denote $\val(\Psi_b) = \max_{x\in \pmo^n} \Psi_b(x)$.
Note that $\val(\Psi_b) \geq 0$ since $\E_x[\Psi_b(x)] = 0$ under a uniformly random $x$.
A $q$-LDC must satisfy $\E_b[\val(\Psi_b)] \geq \Omega(m)$ (see \Cref{obs:lower-bound}); thus we would like to prove that $\E_b[\val(\Psi_b)] \leq o(m)$ for any arbitrary $\cH$ when $k \gg n^{1-2/q}$.

The key connection to Kikuchi graphs is as follows.
For even $q$, $\Psi_b(x)$ is equivalent to the quadratic form of the \emph{signed} Kikuchi matrix $K_b \coloneqq \sum_{i=1}^k b_i K_i$, where $K_i$ is the unsigned Kikuchi matrix of $\cH_i$.
More specifically, $\Psi_b(x) = c\cdot z^\top K_b z$, where $z \in \pmo^{N}$ is some vector of monomials of $x$ and $c>0$ is some scaling.
Then, since $\|z\|_2^2 = N$ for any $x$, it suffices to prove an upper bound on $\|K_b\|_2$.
In particular, we can upper bound $\E_b[\|K_b\|_2] = \E_b\left[ \norm{\sum_{i=1}^k b_i K_i}_2\right]$ using the well-known Matrix Khintchine inequality (\Cref{lem:mk}).

For odd $q$, however, we cannot write $\Psi_b(x)$ directly as a quadratic form, so we need the \emph{Cauchy-Schwarz trick} which we describe next.

\parhead{Cauchy-Schwarz trick:}
Let us first consider a subhypergraph $\cH'$ such that every $C \in \cH'$ contains some tuple $Q \in \binom{[n]}{t}$, where $t \in [q]$.
Let $\Psi_b^{\cH'}$ be the sub-instance of $\Psi_b$ corresponding to $\cH'$.
With slight abuse of notation, for $C \in \cH_i$, we denote $b_C = b_i$.
Then,
\begin{align*}
    \Psi_b^{\cH'}(x)^2 = \parens*{\sum_{C \in \cH'} b_C x_C}^2
    = |\cH'| + \sum_{C \neq C' \in \cH'} b_C b_{C'} x_{C  \oplus C'} \mper
\end{align*}
Notice that $|C \oplus C'| = 2(q-t)$ when $|C \cap C'| = t$.
Recall that the Kikuchi graph for odd $q$ (\Cref{def:kikuchi}) is formed precisely by taking pairs of hyperedges $C,C'$ with intersection $t$ and canceling them.
It is not hard to show that the second term in the above can be expressed as a quadratic form of the signed Kikuchi matrix (see \Cref{lem:upper-bound-via-infty-to-1-norm}), though each edge $S \xleftrightarrow{C,C'} T$ now gets a correlated sign $b_C b_{C'}$.
Note also that since $C \cap C' \neq \varnothing$ and $\cH_1,\dots,\cH_k$ are hypergraph matchings, $C$ and $C'$ must belong to different $\cH_i$'s and thus $b_C$, $b_{C'}$ are different random variables.

\parhead{Pruning heavy vertices:}
In the discussions above, we ignored an important detail: $\E_b[\|K_b\|_2]$ can be large due to some large degree vertices in the Kikuchi graph, i.e., rows/columns in the matrix with large $\ell_1$ norm.
On the other hand, to upper bound $z^\top K_b z$ for $z\in \pmo^N$, it suffices to upper bound the \emph{infinity-to-$1$ norm} defined as $\|K_b\|_{\infty\to1} \coloneqq \max_{u,v\in \pmo^N} u^\top K_b v$.
Then, letting $\wt{K}_b$ be the matrix obtained by removing heavy rows and columns, we can bound $\|K_b\|_{\infty\to1} \leq \|K_b - \wt{K}_b\|_{\infty\to1} + \|\wt{K}_b\|_{\infty\to1}$.\footnote{This is the row pruning step in \cite{GuruswamiKM22}.}
We can thus bound $\E_b \|\wt{K}_b\|_{\infty\to1}$ by its spectral norm using the Matrix Khintchine Inequality (\cref{lem:mk}).

On the other hand, to bound $\|K_b - \wt{K}_b\|_{\infty\to1}$, we need to upper bound the contribution from those pruned rows and columns.
Consider a random $S \in \binom{[n]}{\ell}$ (i.e., a random vertex in the Kikuchi graph), and let $s\in \zo^n$ be the indicator vector of $S$.
For the purpose of this overview, we think of $s$ as a vector with i.i.d.\ Bernoulli variables with bias $\ell/n$.
Then, the random variable $\deg(S)$ can be expressed as a polynomial in $s$ of degree at most $q$.
To bound the tail probabilities of $\deg(S)$, we now invoke a standard polynomial concentration inequality (\Cref{fact:ss}) by Schudy and Sviridenko~\cite{SS12}, which improves upon the classical work of \cite{KV00}.

Crucially, \Cref{fact:ss} requires upper bounds on the \emph{partial derivatives} of the polynomial.
It turns out that these derivatives can be upper bounded by the co-degrees of the hypergraph $\cH$ (see \Cref{lem:derivativebounds})!
Moreover, the co-degree bounds required coincide with those stated in \Cref{lem:even-cover-regular-hypergraph}.

Therefore, our strategy is to use the decomposition described in \Cref{sec:decomposition-overview} on the $q$-XOR instance $\Psi_b$, which gives us sub-instances whose underlying hypergraphs are $t$-approximately strongly regular.
Then, we obtain upper bounds on each sub-instance by using the Cauchy-Schwarz trick (with the good index $t$) and the infinity-to-1 norm of the corresponding Kikuchi matrix, thus finishing the proof.
This is articulated in \Cref{subsec:kikuchireg}.

%% file: prelim.tex
\section{Preliminaries}

\label{sec:Setting}

For a $N \times N$ matrix $A$, we define $\Norm{A}_2 \coloneqq \max_{u \neq 0} \Norm{Au}_2/\Norm{u}_2$ to be its spectral norm and $\Norm{A}_{\infty \rightarrow 1} \coloneqq \max_{u,v \in \on^N} u^{\top} A v$ to be the $\infty \rightarrow 1$ norm. 

\parhead{Hypergraphs} 
In this work, we will work with $q$-uniform hypergraphs $\cH$ over $[n]$. Unless otherwise stated, our hypergraphs are allowed to have parallel or ``multi''-edges.
We write $|\cH|$ to denote the number of hyperedges in $\cH$.

We first recall the standard notion of \emph{co-degrees} in a hypergraph.

\begin{definition}[Co-degrees] \label{def:co-degrees}
  Let $\cH$ be a hypergraph on $[n]$. For any $Q\subseteq[n]$, let $\cH_{|Q}:= \{C\in\cH:Q\subseteq C\}$ be the subset of hyperedges containing $Q$. We define the co-degree of $Q$ in $\cH$ to be $d_{\cH, Q} = \big|\cH_{|Q}\big|$. We let $d_{\cH, t} \coloneqq \max_{Q\in\binom{[n]}{t}}d_{\cH, Q}$ be the maximum possible co-degree of any subset of size $t$.
  We'll refer to $d_{\cH, Q}, d_{\cH, t}$ simply as $d_Q, d_t$ respectively when $\cH$ is clear from the context.
\end{definition}

Observe that $d_1 \geq d_2 \geq \cdots\geq d_q$ and $d_1 \geq |\cH|/n$.

A hypergraph $\cH$ is a matching if all its hyperedges are pairwise disjoint. A proper edge-coloring of a hypergraph $\cH$ is an assignment of colors to the hyperedges of $\cH$ such that all hyperedges incident on any vertex get distinct colors.

\subsection{Locally Decodable Codes}

\begin{definition}[Locally Decodable Code (LDC)]
    A code $E:\pmo^k \to \pmo$ is $(q, \delta, \varepsilon)$-locally decodable if there exists a randomized decoding algorithm $\mathsf{Dec}(\cdot)$ such that $\mathsf{Dec}(\cdot)$ is given oracle access to $y\in\{0, 1\}^n$, takes an $i\in[k]$ as an input, and satisfies the following properties:
    \begin{enumerate}
        \item $\mathsf{Dec}(\cdot)$ makes $\leq q$ queries to the string $y$.
        \item For all $b\in\{0, 1\}^k$, and for all $y\in\{0, 1\}^n$ such that $\operatorname{dist}(y, C(b))\leq\delta n$, $\Pr[\mathsf{Dec}^{y}(i) = b_i]\geq\frac{1}{2} + \varepsilon$.
    \end{enumerate}
    Here $\operatorname{dist}(x, y):= \#\{v\in[n]: x_v\neq y_v\}$ denotes the Hamming distance between $x$ and $y$.
\end{definition}
\begin{definition}[Normally Decodable Codes]\label{def:normalldc}
    A code $E:\pmo^k \to \pmo^n$ is called $(q, \delta, \varepsilon)$-normally decodable if there exist $q$-uniform hypergraph matchings $\cH_1, \cdots, \cH_k$, each with at least $\delta n$ hyperedges,
    such that for every $C\in\cH_i$, we have $\Pr_{b\leftarrow\{\pm 1\}^k}[b_i = \prod_{v\in C} E(b)_v] \geq \frac{1}{2} + \varepsilon$. 
\end{definition}
The following reduction, which is a variant of a classical result of \cite{KT00}, allows us to work with normally decodable codes without loss of generality. 
\begin{fact}[Reduction to Normally Decodable Codes, Lemma 6.2 in \cite{Yekhanin12}] \label{fact:normal-form}
    Let $C:\{0, 1\}^k\to\{0, 1\}^n$ be a $(q, \delta, \varepsilon)$-LDC. Then there is a code $C':\pmo^k\to\pmo^{O(n)}$ that is $(q, \delta', \varepsilon')$-normally decodable, with $\delta'\geq\varepsilon\delta/3q^22^{q - 1}$ and $\varepsilon'\geq\varepsilon/2^{2q}$.
\end{fact}

\subsection{Concentration Inequalities}

The following is a standard concentration result of random matrices.
\begin{fact}[Matrix Khintchine Inequality~\cite{MR355667,MR1150376}; Theorem 4.1.1 of \cite{Tropp15}] \label{lem:mk}
Let $A_1,A_2, \ldots, $  be arbitrary $M \times N$ matrices. Then, 
\[
\E_{b_1, b_2, \ldots \in \on}  \Norm{\sum_i b_i A_i}_2 
\leq \sqrt{2\log (M+N)} \cdot \max \left\{ \Norm{\sum_i A_i A_i^{\top}}^{1/2}_2, \Norm{\sum_i A_i^{\top} A_i}^{1/2}_2\right\} \mcom
\]
where the expectation is taken over independent and uniform $b_i \in \on$. 
\end{fact}
We also utilize the following polynomial concentration inequality.
\begin{fact}[Polynomial Concentration; Theorem 1.1 in \cite{SS12}]
\label{fact:ss}
Let $f(x) = f(x_1, \ldots, x_n)$ be a multilinear polynomial of degree $q$ with non-negative coefficients. Let $Y_1, \ldots, Y_n$ be i.i.d Bernoulli random variables. Write $f(Y):= f(Y_1, \ldots, Y_n)$, and let $\mu_0:= \mathbb{E}f(Y)$. For $r\in[q]$, define 
\[\mu_r:= \max_{Z\in\binom{[n]}{r}}\mathbb{E}_Y\left(\partial_Zf\right)(Y) \mcom \]
where $\partial_Z f:= \left(\prod_{i\in Z}\partial_{x_i}\right)f$. Then there exists an absolute constant $W\geq 1$ such that for any $\lambda > 0$, we have:
\[\Pr(|f(Y)-\mu|\geq\lambda)\leq e^2\cdot\max\left\{\max_{r\in[q]}\exp\left(-\frac{\lambda^2}{\mu_0\mu_rW^q}\right), \max_{r\in[q]}\exp\left(-\left(\frac{\lambda}{\mu_rW^q}\right)^{1/r}\right)\right\} \mper \]
\end{fact}

%% file: lowerbound.tex
\section{Proof of the Main Theorem}

In this section, we introduce the core components of our argument and use them to prove:

\begin{theorem}[Main Theorem] \label{thm:main-detail}
Let $E:\on^k \rightarrow \on^n$ be a $(q, \delta, \epsilon)$-LDC for any $q\geq 3$. Then, 
\[
k \leq n^{1-\frac{2}{q}} \cdot  \frac{2^{O(q)} \log^{4}n}{\epsilon^{6} \delta^{2}} \mper
\]
\end{theorem}

Given \Cref{fact:normal-form}, we can assume that the LDC $E:\on^k \rightarrow \on^n$ is in the normal form. Let $\cH_1, \cH_2, \ldots, \cH_k$ be the $q$-uniform matchings of size $\geq \delta n$ for the associated normal decoder. By a slight abuse of notation, we will use $\cH$ to denote both the tuple $(\cH_1, \cH_2, \ldots, \cH_k)$ and the hypergraph formed by their union $\bigcup_{i =1}^k \cH_i$. As in~\cite{AlrabiahGKM23}, we will associate a system of $q$-XOR constraints that encode the success of the decoder for $E$. Let us first define this $q$-XOR system.

\begin{definition}[$q$-XOR Polynomial for a Normal LDC]
Given a collection of $q$-uniform matchings $\cH=(\cH_1, \cH_2, \ldots, \cH_k)$ on $[n]$ and any $b\in \pmo^k$, let 
\[
\Psi^{\cH}_b(x) = \sum_{i = 1}^k b_i \sum_{C \in \cH_i} x_C\mcom
\]
be a degree-$q$ polynomial in variables $x_1, x_2, \ldots, x_n$ corresponding to the codeword bits. We drop the superscript and write $\Psi_b$ when the collection of $q$-uniform matchings $\cH$ is clear from the context. Let $\val(\Psi_b) = \max_{x \in \on^n} \Psi_b(x)$ be the \emph{value} of $\Psi_b$.  
\end{definition}

Informally speaking, on average over $b\in \on^k$, each monomial in $\Psi_b$ corresponds to a local decoding constraint that is satisfied with a non-trivial advantage by the codeword $E(b)$. This is immediate given \cref{def:normalldc}, and we record it as the following observation: 

\begin{observation} \label{obs:lower-bound}
For a $(q,\delta,\eps)$-normally decodable code $E$, we have
$\E_{b\in\{\pm 1\}^k}[\Psi_b(E(b))] \geq 2\varepsilon |\cH|$.
\end{observation}
\begin{remark}
    Note that for a normally decodable code, $\cH$ is a union of $k$ hypergraph matchings and since $k \leq n$, the hypergraph thus has size $\leq nk\leq n^2$.
\end{remark}

Our goal in the rest of this section is to prove an upper bound on $\E_{b}[\val(\Psi_b)]$, which, when combined with \cref{obs:lower-bound}, will give us \cref{thm:main-detail}. To do this, we will first identify a sufficient condition on $\cH$ called \emph{approximate strong regularity} (\Cref{def:approxstrongregdef}), and argue that this condition can be ensured without the loss of generality by an appropriate decomposition applied to an arbitrary hypergraph $\cH$.

\paragraph{Approximate Strong Regularity} Our notion of approximate strong regularity can be seen as a generalization of similar definitions used in in~\cite{GuruswamiKM22,HsiehKM23}. Their notions of regularity demand absolute upper bounds on all co-degrees in the hypergraph. In contrast, our definition below identifies a special index $t$ and demands 1) that all size $t$ sets have their co-degrees close up to a factor $2$ to a fixed value $d_t$, and 2) the co-degree of sets of size $r\neq t$ be appropriately small \emph{relative to} $d_t$. Crucially, our definition neither demands nor guarantees any non-trivial absolute upper bound on any co-degree. The following definition (discussed earlier in \Cref{sec:overview}) notes the relative upper bounds we need so as to make a crucial technical component of our argument go through.

\restatedefinition{def:goodindex}
We can now define approximate strong regularity. Informally speaking, $\cH$ is approximately regular if one can partition the hyperedges of $\cH$ into groups so that 1) the hyperedges in a group all contain some set of size $t$, 2) each group has the same size up to a factor $2$, and, for the co-degree tuple formed by $d_r$ being the maximum possible co-degree of any set of size $r$, $t$ is a good index.

\restatedefinition{def:approxstrongregdef}

\begin{remark}
It is easy to see that we have $p_t\leq\frac{2|\cH|}{d_{t}}$. Also, let $P=\bigsqcup_{t\in[q]} P_t,$ where $P_t$ is the set of all indices $\pi$ such that $\cH^\pi$ is $t$-approximately strongly regular.
\end{remark}

Our key technical argument is captured in the following lemma that shows an upper bound on $ \E_b\left[\val(\Psi_b^{\cH})^2\right]$ if $\cH$ is $t$-approximately strongly regular.

\begin{lemma}[] \label{lem:refute-approx-regular}
Let $\cH_1, \cH_2, \ldots, \cH_k$ be $q$-uniform matchings on $[n]$ such that $\bigcup_{i\in [k]}\cH_i$ is $t$-approximately strongly regular with degree sequence $d_1 \geq d_2 \geq \cdots\geq d_q$ such that $k \geq d_1 \geq \ell \coloneqq n^{1-2/q}\log n$. Then we have
\[
\E_b\left[ \val\left(\Psi^\cH_b\right)^2\right] \leq \frac{2|\cH|^2}{d_t} + |\cH| n \sqrt{k\ell}  O(\log n)^{1/2}\mper 
\]
\end{lemma}

We will prove \Cref{lem:refute-approx-regular} in \Cref{sec:refute-approx-regular} modulo a few technical results proved in subsequent sections. To finish the proof, we will show in \Cref{decompsec} that every $\cH = (\cH_1, \cH_2, \ldots, \cH_k)$ can be decomposed into approximately strongly regular pieces up to a negligible error.

\begin{lemma}
\torestate{
    \label{lem:new_bucketing_arpon_version}
    Let $\cH$ be a $q$-uniform hypergraph on $[n]$. Let $\eta>0$. Then, there is a collection $\{\mathcal{H}^\pi : \pi\in P\}$, with $|P|\leq q\lceil\log_2|\cH|\rceil + 1$, of hypergraphs such that:
\begin{enumerate}
    \item For every $\pi\in P$, there exists $t\in[q]$ such that $\cH^{\pi}$ is $t$-approximately strongly regular.
    \item $\bigsqcup_{\pi\in P} \mathcal{H}^{\pi}\subseteq \mathcal{H}$,
    \item $\left| \bigsqcup_{\pi\in P} \mathcal{H}^{\pi} \right|\geq (1-\eta)|\cH|$,
    \item $|\cH^\pi| \geq \frac{\eta|\cH|}{q\lceil\log_2|\cH|\rceil + 1}$ for all $\pi\in P$.
\end{enumerate}
}
\end{lemma}

Let us complete the proof of \Cref{thm:main-detail} using \Cref{lem:refute-approx-regular,lem:new_bucketing_arpon_version}.

\begin{proof}[Proof of \Cref{thm:main-detail}]
Let $E:\on^k \rightarrow \on^n$ be a $(q,\delta,\epsilon)$-LDC. Let $\epsilon' = \epsilon/2^{2q}$ and $\delta' = \epsilon \delta/3q^2 2^{q-1}$ be the parameters obtained from \Cref{fact:normal-form} when we reduce to the normal form. From \Cref{obs:lower-bound}, we thus have that: $\E_{b \sim \on^k}\left[ \val(\Psi_b)\right] \geq 2 \epsilon' |\cH|$.
  
Set $\ell = n^{1-\frac{2}{q}}\log n$.
Suppose towards a contradiction that $k \geq An^{1-\frac{2}{q}} \cdot\frac{q^2\log^{4} n}{\epsilon'^{4} \delta'^{2}}$ for a large enough constant $A > 0$.
Note that this expression is at most $n^{1-\frac{2}{q}} \cdot \frac{2^{O(q)}\log^4 n}{\eps^6 \delta^2}$.

Let us apply the results we developed so far to construct an upper bound on $\E_b[ \val(\Psi_b)]$. 
First, we apply the decomposition from \Cref{lem:new_bucketing_arpon_version} for $\eta = \epsilon'/2$ and get a collection of subhypergraphs $\{\cH^{\pi}: \pi \in P\}$.
Here, we have $|P| \leq O(q\log n)$ and $\sum_{\pi\in P}|\cH^{\pi}| \geq (1-\eta) |\cH|$.
Let $P_t$ be the set of all indices $\pi$ such that $\cH^\pi$ is $t$-approximately strongly regular (where $t$ is the good index).
For simplicity, we denote $d_r^{\pi} \coloneqq d_{\cH^{\pi},r}$ to be the co-degree bounds of $\cH^{\pi}$.
Then we can write:
\begin{align*}
    \E_b[ \val (\Psi_b)] 
    &\leq \E_b\left[ \val \parens*{ \sum_{\pi \in P} \Psi^{\cH^{\pi}}_b } \right] + \eta |\cH| 
    \leq \sum_{\pi \in P} \E_b\left[ \val\left(\Psi^{\cH^{\pi}}_b\right)\right] + \eta |\cH| \\
    &\leq \sqrt{O(q\log n) \cdot \sum_{\pi \in P} \E_b\left[ \val\left(\Psi^{\cH^{\pi}}_b\right)^2\right]} + \eta |\cH| \mcom
    \numberthis \label{eq:val-Psi-bound}
\end{align*}
where the last inequality follows from Cauchy-Schwarz.

For each $\pi\in P$, we have $|\cH^{\pi}| \geq \frac{\eta|\cH|}{q\log_2|\cH| + 1} \geq \frac{\eta \delta' k n}{3q\log n}$ since $|\cH|\leq n^2$.
Thus, we have $d_1^{\pi} \geq |\cH^{\pi}|/n \geq \frac{\eta \delta' k}{3q\log n} \gg \ell$ by the lower bound on $k$.
This allows us to invoke \Cref{lem:refute-approx-regular} (which requires $d_1^{\pi} \geq \ell$) for all $\cH^{\pi}$ for $\pi\in P_t$:
\begin{align*}
    \E_b\left[ \val\left(\Psi^{\cH^{\pi}}_b\right)^2\right]
    \leq \frac{2|\cH^{\pi}|^2}{d_t^\pi} + O(1) \cdot n \sqrt{k\ell}  (\log n)^{1/2} |\cH^{\pi}| \mper
\end{align*}
Since $t$ is a good index for $\cH^{\pi}$, condition (\ref{item:d-t-lower-bound}) of \Cref{def:goodindex} gives that $d_t^{\pi} \geq d_1^{\pi} n^{-1+2/q} \geq |\cH^{\pi}| n^{-2+2/q}$.
Thus, the first term above is upper bounded by $2|\cH^{\pi}| n^{2-2/q}$.
Thus, summing over $\pi\in P_t$ and $t\in[q]$, we have
\begin{align*}
    \sum_{\pi\in P} \E_b\left[ \val\left(\Psi^{\cH^{\pi}}_b\right)^2\right]
    &\leq  O(1) \cdot \abs*{\cH} \parens*{ n^{2-\frac{2}{q}} + n\sqrt{k\ell \log n} } \mper
\end{align*}
We next show that the above is upper bounded by $\frac{c\eps'^2}{q\log n} \cdot |\cH|^2$ for some small enough constant $c$.

Note that $|\cH| \geq \delta' kn$.
Clearly, we have $n^{2-\frac{2}{q}} \leq o\big(\frac{\eps'^2}{q\log n}\big) \cdot |\cH|$ by our choice of $k$.
Next, we verify that $n\sqrt{k\ell \log n} \leq O\big(\frac{\eps'^2}{q\log n}\big) \cdot \delta'kn$.
This is equivalent to $k \geq \Omega\big(\frac{q^2}{\eps'^4 \delta'^2}\big) \ell \log^3 n$.
Since $\ell = n^{1-\frac{2}{q}} \log n$, this is satisfied by our choice of $k$.

Therefore, by \Cref{eq:val-Psi-bound} we have $\E_b[ \val(\Psi_b)] \leq (\eps'/2 + \eta) |\cH| \leq \eps'|\cH|$.
This contradicts that $\E_b[ \val(\Psi_b)] \geq 2\eps' |\cH|$.
\qedhere

\end{proof}

\subsection{Proof of \texorpdfstring{\Cref{lem:refute-approx-regular}}{Lemma~\ref{lem:refute-approx-regular}}: refuting regular \texorpdfstring{$q$}{q}-XOR} \label{sec:refute-approx-regular}
In this section, we prove \Cref{lem:refute-approx-regular}. To do this, we would like to write $\Psi_b$ as a quadratic form of a ``Kikuchi'' matrix and then upper bound it using the spectral norm of the matrix. As in prior works, the main issue is that $\Psi_b(x)$ is an odd-degree homogeneous polynomial while Kikuchi matrices naturally yield even-degree polynomials as their quadratic forms. Therefore, we will upper bound $\val(\Psi_b)^2$ and use the Cauchy-Schwarz trick as described in \Cref{sec:nonlinear-overview}. This yields an even-degree polynomial, which we will then analyze a spectral upper bound for. We first define this even-degree polynomial and then show how it can be written as a quadratic form of an appropriate Kikuchi matrix. In order to ensure independence of some matrices that arise later in our analysis, we will define our polynomial with respect to a fixed partition $L \sqcup R$ of $[k]$.

\begin{definition}\label{def:spectralpolydef}

    Let $\cH$ be a $t$-approximately strongly regular hypergraph. Let $\{\cH^{(t)}_\theta\}_{\theta \in[p_t]}$ be the partition given by approximate strong regularity, so that for $\theta \in [p_t]$, there is a $Q_{\theta} \in \binom{[n]}{t}$ such that every $C \in \cH_{\theta}$ satisfies $Q_{\theta} \subseteq C$. For $b\in \pmo^k$ and a partition $L\sqcup R=[k]$, we define
    \[
    f^\cH_{L,R,t}(x) := \sum_{\theta\in [p_t]} \sum_{\substack{i\in L\\j\in R}}b_ib_j \sum_{\substack{C\in \cH^{(t)}_\theta \cap \cH_i,\\ C'\in \cH^{(t)}_{\theta} \cap \cH_j}} x_{ C\setminus Q_{\theta}}x_{C'\setminus Q_{\theta}} \mper
    \]
\end{definition}

\begin{remark}[Unraveling the definition of $f^\cH_{L,R,t}$]
To define $f^\cH_{L,R,t}$, we work with an arbitrary partition of the original matchings $\cH_1, \cH_2, \ldots, \cH_k$ into a ``left'' group $L$ and a ``right'' group $R$ (this partitioning later ensures independence when we apply matrix concentration).

Then, $f^\cH_{L,R,t}$ can be viewed as merging pairs of hyperedges between the left and right groups within each $\cH_{\theta}^{(t)}$.
That is, $f^\cH_{L,R,t}$ includes one monomial for each pair $C,C' \in \cH^{(t)}_{\theta}$ such that $C \in \cH_i$ for $i \in L$ and $C' \in \cH_j$ for $j \in R$, obtained by multiplying $b_i x_C$ and $b_j x_{C'}$.  Since  $Q_{\theta} \subseteq C,C'$, the product has a squared term $x_{Q_{\theta}}^2$ which ``cancels'' out to $1$.
\end{remark}

\paragraph{The Cauchy-Schwarz trick:} Let us relate $\E_{L,R} \big[ \val \big(f^\cH_{L,R,t}\big)\big]$ with $\Psi^\cH_b$ and conclude that it is enough to bound $\E_{L,R} \big[ \val \big(f^\cH_{L,R,t}\big)\big]$. Here, the expectation is over a uniformly random partition $L \cup R$ of $[k]$.
We emphasize again that $t$-approximate strong regularity (\Cref{def:approxstrongregdef}) guarantees a partitioning $\{\cH_{\theta}^{(t)}\}_{\theta\in[p_t]}$.
\begin{lemma}[Cauchy-Schwarz Trick]\label{lem:cstrickNEW}
Let $\cH$ be a $t$-approximately strongly regular hypergraph. Then we have
    \[ \val(\Psi^\cH_b)^2 \leq \frac{2|\cH|^2}{d_t} + \frac{8|\cH|}{d_t}\E_{L,R}\left[\val\big(f^\cH_{L,R,t}\big)\right] \mper \]
\end{lemma}
\begin{proof}
For convenience, we will slightly abuse notation and denote $b_C$ to be $b_i$ for $C \in \cH_i$.
We have:
\begin{align*}
\Psi^\cH_b(x) &= \sum_{C\in \cH} b_C x_C
= \sum_{\theta \in [p_t]} x_{Q_{\theta}} \sum_{C \in \cH_{\theta}^{(t)}} b_C x_{C \setminus Q_{\theta}} \mper
\end{align*}
Thus, by the Cauchy-Schwarz inequality:
\begin{align*}
\Psi^\cH_b(x)^2 &\leq \left(\sum_{\theta \in [p_t]} x_{Q_{\theta}}^2\right) \sum_{\theta \in [p_t]} \left( \sum_{C \in \cH_{\theta}^{(t)}} b_C x_{C \setminus Q_{\theta}} \right)^2 
= p_t \sum_{\theta\in[p_t]} \parens*{ \abs*{\cH_{\theta}^{(t)}} + \sum_{C \neq C' \in \cH_{\theta}^{(t)}} b_C b_{C'} x_{C\setminus Q_{\theta}} x_{C' \setminus Q_{\theta}} } \mper
\end{align*}
Note that all hyperedges in $\cH_{\theta}^{(t)}$ contain $Q_{\theta}$, so they must belong to distinct matchings.
For any $C \neq C' \in \cH_{\theta}$ where $C \in \cH_i$ and $C' \in \cH_j$ for some $i\neq j \in [k]$, when we take a random partition $L \sqcup R = [k]$, we have that $i\in L$ and $j\in R$ with probability exactly $1/4$.
Thus, the above expression equals
\begin{align*}
    p_t |\cH| + p_t \cdot 4 \E_{L,R} \sum_{\theta \in [p_t]} \sum_{i\in L, j\in R} b_i b_j \sum_{\substack{C \in \cH_i \cap \cH^{(t)}_{\theta},\\ C' \in \cH_j \cap \cH^{(t)}_{\theta}}}  x_{C\setminus Q_{\theta}} x_{C' \setminus Q_{\theta}}
    = p_t \parens*{|\cH| + 4 \E_{L,R}\left[f^\cH_{L,R,t}(x)\right]} \mper
\end{align*}
Here, the last inequality uses our definition of $f^\cH_{L,R,t}$ in \Cref{def:spectralpolydef}.
Thus, since we have $p_t\leq \frac{2|\cH|}{d_t}$ by the remark following \cref{def:approxstrongregdef}, we apply Jensen's Inequality to get
\[ \val(\Psi_b)^2 \leq \frac{2|\cH|^2}{d_t} + \max_{x\in\{\pm 1\}^n} \frac{8|\cH|}{d_t}\E_{L,R}\left[f^\cH_{L,R,t}(x)\right]\leq \frac{2|\cH|^2}{d_t} + \frac{8|\cH|}{d_t}\E_{L,R}\left[f^\cH_{L,R,t}(x)\right] \mper \qedhere \]
\end{proof}

\paragraph{$f^\cH_{L,R,t}$ as a quadratic form} We will write $f^\cH_{L,R,t}$ as a quadratic form of a natural (signed) \emph{Kikuchi} matrix (same as the one appearing in~\cite{HsiehKM23}).
See \Cref{fig:KikuchiGraph} for an example.
For any set $S$, and any $X, Y\subseteq S$, define $X\oplus Y:= (X\setminus Y)\cup(Y\setminus X)$ to be the symmetric difference of $X$ and $Y$. We define Kikuchi graphs associated with the monomials appearing in $f^\cH_{L,R,t}$. 

\begin{definition}[Level-$\ell$ Kikuchi Graph] \label{def:kikuchi}
Given a $t$-approximately strongly regular hypergraph $\cH$, hypergraph matchings $\cH = (\cH_1,\dots,\cH_k)$, and a partition $\{\cH_{\theta}^{(t)}\}_{\theta\in [p_t]}$ from approximate strong regularity, we define the Kikuchi graph as follows. 

We create two copies $(x_u,1)$ and $(x_u,2)$ of the original $n$ variables $x_u$ for each $1 \leq u \leq n$. The Kikuchi graph has vertices indexed by subsets $S \in {[n] \times [2] \choose \ell}$ for a parameter $\ell \in \N$.
For any $C \neq C' \in \cH^{(t)}_{\theta}$, let  $\widetilde{C} \coloneqq C\setminus Q_{\theta}$ and $\widetilde{C}' \coloneqq C'\setminus Q_{\theta}$.
For any $S,T \in {[n] \times [2]\choose \ell}$, we say that $S\xleftrightarrow{C, C'} T$ if:
    \begin{enumerate}
        \item $S\oplus T = X\oplus Y$, for $X \coloneqq \widetilde{C}\times\{1\}, Y \coloneqq \widetilde{C}'\times\{2\}$. Clearly, $X\oplus Y = X\cup Y$ as $X, Y$ are disjoint.
        \item When $q - t$ is even, $|S\cap X| = |S\cap Y| = |T\cap X| = |T\cap Y| = (q - t)/2$.
    \item When $q - t$ is odd, either
    \begin{enumerate}
        \item $|S\cap X| = |T\cap Y| = \left\lceil\frac{q - t}{2}\right\rceil, |S\cap Y| = |T\cap X| = \left\lfloor\frac{q - t}{2}\right\rfloor$ or
        \item $|S\cap X| = |T\cap Y| = \left\lfloor\frac{q - t}{2}\right\rfloor, |S\cap Y| = |T\cap X| = \left\lceil\frac{q - t}{2}\right\rceil$.
    \end{enumerate}
    \end{enumerate}
\end{definition}

\begin{figure}[ht]
    \centering
    \includegraphics[width=0.75\linewidth]{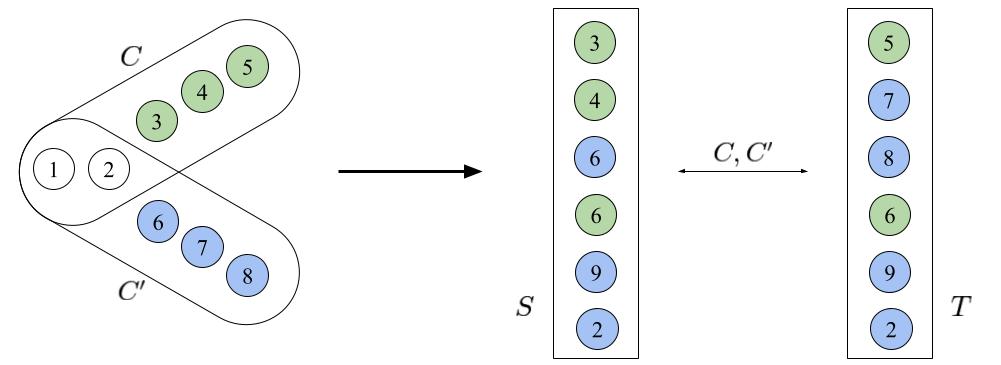}
    \caption{In this example, the hyperedges $C, C'$ induce an edge between $S, T$. Here, $q = 5, t = 2$. Vertices of the form $(x, 1)$ are colored green, while vertices of the form $(x, 2)$ are colored blue. Note that $S\oplus T = (C\setminus Q_\theta)\times\{1\}\bigcup\ (C'\setminus Q_\theta)\times\{2\}$. Also note how $S$ contains $2 = \left\lceil\frac{q - t}{2}\right\rceil$ green elements and $1 = \left\lfloor\frac{q - t}{2}\right\rfloor$ blue element from $S\oplus T$, while $T$ contains $1 = \left\lfloor\frac{q - t}{2}\right\rfloor$ green element and $2 = \left\lceil\frac{q - t}{2}\right\rceil$ blue elements from $S\oplus T$.}
    \label{fig:KikuchiGraph}
\end{figure}
Notice that each pair $(C,C')$ contributes a matching (at most one edge per vertex) in the Kikuchi graph. The following result gives the size of this matching.
\begin{fact}[Observation 3.7 of \cite{HsiehKM23}] \label{fact:bound-D}
    The Kikuchi graph defined in \Cref{def:kikuchi} has $N = \binom{2n}{\ell}$ vertices, and each pair $(C,C')$ contributes a $D$-sized matching in the Kikuchi graph, where
    \begin{equation*}
        D \coloneqq \binom{q-t}{\floor{\frac{q-t}{2}}}^2  \binom{2n-2(q-t)}{\ell-(q-t)}  \cdot 2^{\1(\text{$q-t$ is odd})}
        \geq N \parens*{\frac{\ell}{2n}}^{q-t} \mper
    \end{equation*}
\end{fact}

We next define the \emph{signed} Kikuchi matrix according to $b\in \pmo^k$.
We will use $\ol{K}$ to denote unsigned matrices and $K$ to denote signed matrices according to $b$.

\begin{definition}[Signed Kikuchi Matrix] \label{def:signed-kikuchi}
    Assume the same setting as \Cref{def:kikuchi}.
    Let $N \coloneqq \binom{2n}{\ell}$.
    Let $\ol{K}^{(C,C')} \in \{0,1\}^{N\times N}$ be the matrix with rows/columns indexed by ${[n] \times [2] \choose \ell}$ and $\ol{K}^{(C,C')}(S,T) = \mathbf{1}\{S\overset{C, C'}{\longleftrightarrow}T\}$, i.e., the adjacency matrix of the matching formed by $(C,C')$.
    For $i\neq j\in [k]$, define the unsigned matrix
    \begin{align*}
        \ol{K}_{i,j,t} \coloneqq \sum_{\theta \in [p_t]} \sum_{\substack{C\in\cH_i \cap \cH^{(t)}_{\theta},\\C'\in \cH_j\cap \cH^{(t)}_{\theta}}} \ol{K}^{(C,C')} \mper
    \end{align*}
    Next, for $b\in \pmo^k$ and a partition $L \sqcup R = [k]$, we define the signed matrix
    \begin{align*}
        K_{i,t} \coloneqq \sum_{j \in R} b_j \ol{K}_{i,j,t} \mper
    \end{align*}
    Here, we omit the dependence on $b$, $L$ and $R$ for simplicity.
\end{definition}

We now observe that $f^\cH_{L,R,t}$ is a quadratic form of $\sum_{i \in L} b_i K_{i,t}$. 

\begin{lemma} \label{lem:upper-bound-via-infty-to-1-norm}
For any $x \in \on^n$, let $z \in \on^{N}$ for $N= \binom{2n}{\ell}$ be defined so that for any $S=S_1\sqcup S_2$ where $|S|=\ell,$ $S_1\subseteq [n]\times\{1\},$ and $S_2\subseteq [n]\times\{2\},$ we have $z_S:=\prod_{u\in S_1}x_u\prod_{u\in S_2} x_u.$ Then,
\[
z^{\top} \left(\sum_{i \in L} b_i K_{i,t}\right) z = D \cdot  f^\cH_{L,R,t}(x) \mper
\]
Thus, $\val(f^\cH_{L,R,t}) \leq D^{-1} \Norm{\sum_{i \in L} K_{i,t}}_{\infty \rightarrow 1}$.
\end{lemma}
\begin{proof}
From \Cref{fact:bound-D}, we have that $z^{\top} \ol{K}^{(C,C')}z = D \cdot x_{C \setminus Q_{\theta}} x_{C' \setminus Q_{\theta}}$. Thus,
\[
D \cdot f^\cH_{L,R,t}(x) = \sum_{i\in L, j\in R} b_i b_j \cdot z^{\top} \ol{K}_{i,j,t} z = z^{\top} \left(\sum_{i \in L} b_i K_{i,t}\right)z \leq \Norm{\sum_{i \in L} b_i K_{i,t}}_{\infty\rightarrow 1}\mper\qedhere
\]
\end{proof}

\paragraph{Spectral bound via Matrix Khintchine} We will use the Matrix Khintchine inequality (\Cref{lem:mk}) to upper bound the spectral norm of $\sum_{i \in L} b_i K_{i,t}$.
The left-right partitioning scheme ensures that the $K_{i,t}$'s are independent random matrices as $i$ varies in $L$. Recall that $d_r = \max_{Q\in \binom{[n]}{r}} d_{\cH, Q}$ is the maximum co-degree of size-$r$ subsets.
The following key technical result shows that one can upper bound the maximum degree of $K_{i,t}$ according to $d_t$ for all but a negligible fraction of vertices in $K_{i,t}$ for any $i \in L$.

\begin{lemma}\label{lem:kikuchiregularity}
    Assume the same setting as \cref{def:kikuchi} and suppose that $k \geq d_1 \geq \ell = n^{1-2/q}\log n$. 
    Then $d_t \geq d_1 n^{-2(t-1)/q}$ if $t < q/2$ and $d_t \geq d_1 n^{-1+2/q}$ if $t > q/2$.
    Moreover, for every partition $L \sqcup R = [k]$ and every $K_{i,t}$ for $i \in L$, the following holds: 
    there exists a constant $W$ such that for all but a $n^{-100q}$ fraction of vertices $S\in{[n]\times [2]\choose \ell}=V(K_{i,t}),$ we have
    \[\deg_{K_{i,t}}(S)\leq \left(\frac{\ell}{n}\right)^{q-t} n d_t \cdot W^q\mcom\]
    where $\deg_{K_{i,t}}(S)$ denotes the maximum $\ell_1$ norm of row $S$ in $K_{i,t}$ over all $b_j$ for $j \in R$.
\end{lemma}

Note that the lower bound on $d_t$ in \Cref{lem:kikuchiregularity} and $d_1 \geq n^{1-2/q}$ ensure that $(\frac{\ell}{n})^{q-t}nd_t \geq 1$.

We will prove \Cref{lem:kikuchiregularity} in \Cref{subsec:kikuchireg}, using a key technical lemma (\Cref{lem:good-index-exists}) that guarantees a ``good'' index $t$.
Having the above $\ell_1$ norm upper bound for most rows in $K_{i,t}$ allows us to easily obtain an upper bound on $\Norm{\sum_{i \in L} b_i K_{i,t}}_{\infty \to \1}$ using Matrix Khintchine (after removing the bad rows and columns).

\begin{lemma} \label{lem:infty-to-1-norm-bound}
We have $D^{-1} \Norm{\sum_{i \in L} b_i K_{i,t}}_{\infty\rightarrow 1} \leq \sqrt{k\ell} nd_t O(\log n)^{1/2}$.
\end{lemma}
\begin{proof}
Fix any partition $L \sqcup R = [k]$.
For every $i \in L$, let $\tilde{K}_{i,t}$ be obtained by zeroing out any row and column indexed by $S$ such that $\deg_{K_{i,t}}(S) > (\ell/n)^{q-t}nd_t \cdot W^q$.
Then, for any $i \in L$, $K_{i,t}$ and $\tilde{K}_{i,t}$ differ in at most $N \cdot k (n^{-100q} + e^{-\ell/4})$ rows and columns (by \cref{lem:kikuchiregularity}). Further, observe that the $\ell_1$ norm of any row of $K_{i,t}$ is at most $|\cH|^2 \ell^{q} \leq n^{3q}$. As a result,
\begin{equation} \label{eq:khintchine-app}
    \Norm{\sum_{i \in L} b_i K_{i,t}}_{\infty \to \1} \leq \Norm{\sum_{i \in L} b_i \tilde{K}_{i,t}}_{\infty \to \1} + 2N \left(n^{-100q} + e^{-\ell/4}\right) \cdot n^{3q}\leq \Norm{\sum_{i \in L} b_i \tilde{K}_{i,t}}_{\infty \to \1} + o(N) \mcom
\end{equation}
where we use the naive $\ell_1$ norm bound for the rows and columns removed from $K_{i,t}$.

Next, observe that $\Norm{\sum_{i \in L} b_i \tilde{K}_{i,t}}_{\infty \to \1} \leq N \Norm{\sum_{i \in L} b_i \tilde{K}_{i,t}}_{2}$. Applying Matrix Khintchine inequality (\Cref{lem:mk}) and the triangle inequality:
\begin{align*} 
\E_{b \in \on^k} \left[ \Norm{\sum_{i \in L} b_i \tilde{K}_{i,t}}_{2}\right] \leq \sigma \sqrt{2\log N} \mcom
\text{ where } \sigma^2 = \Norm{\sum_{i\in L} \tilde{K}_{i,t}^2}_2 \leq |L| \cdot \max_{i\in L} \Norm{\wt{K}_{i,t}}_2^2 \mper
\end{align*}

Now, since $\tilde{K}_{i,t}$ has small degrees, we can upper bound its spectral norm by the largest $\ell_1$ norm of any row, i.e., $\max_{S} \deg_{\tilde{K}_{i,t}}(S)$.
So, we have $\sigma \leq \sqrt{k} \cdot (\ell/n)^{q-t} nd_t \cdot W^q$. 

Plugging back in \eqref{eq:khintchine-app}, we have:
\begin{equation} \label{eq:intermediate-bound}
\E_{b \in \on^k} \left[ \Norm{\sum_{i \in L} b_i K_{i,t}}_{\infty \to \1}\right] \leq O(N\sqrt{\log N}) \sqrt{k} \left(\frac{\ell}{n}\right)^{q-t} n d_t \cdot W^q + o(N) \mper
\end{equation}
We verify that the first term dominates the second above.
By assumption, $k \geq d_1 \geq \ell = n^{1-2/q}\log n$, and by \Cref{lem:kikuchiregularity}, we have $d_t \geq d_1 n^{-2(t-1)/q}$, thus
$\sqrt{k\ell} \big(\frac{\ell}{n}\big)^{q-t} \cdot nd_t \geq d_t n^{2(t-1)/q} \geq 1$.
Then, using \Cref{fact:bound-D}, we have $N/D \leq (2n/\ell)^{q-t}$, and $\log N \leq O(\ell \log n)$.
Thus, multiplying \eqref{eq:intermediate-bound} by $D^{-1}$ gives the desired upper bound $\sqrt{k\ell} nd_t O(\log n)^{1/2}$.
\end{proof}

With the Cauchy-Schwarz trick (\Cref{lem:cstrickNEW}) and the bound on $\Norm{\sum_{i \in L} b_i K_{i,t}}_{\infty\rightarrow 1}$ (\Cref{lem:infty-to-1-norm-bound}), \Cref{lem:refute-approx-regular} is almost an immediate corollary.

\begin{proof}[Proof of \Cref{lem:refute-approx-regular}]
Using \Cref{lem:cstrickNEW}, we have that 
\begin{equation} \label{eq:recall-CS}
\val\left(\Psi_b\right)^2 \leq \frac{2|\cH|^2}{d_t} + \frac{8|\cH|}{d_t} \E_{L,R} \left[\val\left(f^\cH_{L,R,t}\right)\right]\mper
\end{equation}
By \Cref{lem:upper-bound-via-infty-to-1-norm}, we can write $f^\cH_{L,R,t}(x)$ as $D^{-1}$ times a quadratic form of the signed Kikuchi matrix $\sum_{i\in L} b_i K_{i,t}$ from \Cref{def:signed-kikuchi}.
Then, by \Cref{lem:infty-to-1-norm-bound},
\[
\E_{L,R} \E_b[ \val(f_{L,R,t})] \leq \frac{1}{D} \E_{L,R} \E_b\left[\Norm{\sum_{i \in L}b_i K_{i,t}}_{\infty \to \1}\right]
\leq \sqrt{k\ell} \cdot nd_t O(\log n)^{1/2} \mper
\]
Plugging back this bound in \eqref{eq:recall-CS} completes the proof.
\end{proof}

\subsection{Proof of \texorpdfstring{\Cref{lem:kikuchiregularity}: bounding heavy vertices}{Lemma~\ref{lem:kikuchiregularity}}} \label{subsec:kikuchireg}
In this section, we prove \Cref{lem:kikuchiregularity}. First, we observe that the degree of any vertex $S$ of $K_{i,t}$ can be upper bounded by a polynomial $\Deg$ in the $0$-$1$ indicator of the vertex $S$. We then use a polynomial concentration inequality (\Cref{fact:ss}) to analyze the concentration of $\Deg$ as a function of a uniformly random $s$. The concentration inequality relies on expected partial derivatives of $\Deg$ that have natural combinatorial interpretations and can be described as a function of the co-degree sequence $d_1 \geq d_2 \geq \cdots \geq d_q$ of the input hypergraph $\cH = (\cH_1, \cH_2, \ldots, \cH_k)$.
In particular, the required bounds are precisely those such that $d_t$ is a good index as defined in the technical overview (\Cref{sec:regular-hypergraphs}), which we restate below.
\restatedefinition{def:goodindex}

We now prove \Cref{lem:kikuchiregularity}.

\begin{definition}[The Degree Polynomial]
Assume the same setting as \cref{def:kikuchi}, and fix an $i\in [k]$.
Let $s,s' \in \zo^{n}$ be indicator vectors of the sets $S \cap ([n] \times \{1\})$ and $S \cap ([n] \times \{2\})$, respectively. Consider the polynomial $\Deg(s,s')$ defined so that

\[
\Deg(s,s') = \sum_{\theta \in [p_t]} \sum_{\substack{C \in \cH_i \cap \cH^{(t)}_{\theta}, \\C' \neq C \in \cH^{(t)}_{\theta}}} \sum_{\substack{R \in {C\setminus Q_{\theta} \choose (q-t)/2},\\ R' \in {C' \setminus Q_{\theta} \choose (q-t)/2}}} s_R s'_{R'} \mcom
\]
if $t$ is even, and, 
\[
\Deg(s,s') = \sum_{\theta \in [p_t]} \sum_{\substack{C \in \cH_i \cap \cH^{(t)}_{\theta},\\ C' \neq C \in \cH^{(t)}_{\theta}}} \left( \sum_{\substack{R \in {C\setminus Q_{\theta} \choose \lceil (q-t)/2 \rceil}, \\R' \in {C' \setminus Q_{\theta} \choose \lfloor (q-t)/2 \rfloor}}} s_R s'_{R'} + \sum_{\substack{R \in {C\setminus Q_{\theta} \choose \lfloor (q-t)/2 \rfloor},\\ R' \in {C' \setminus Q_{\theta} \choose \lceil (q-t)/2 \rceil}}} s_R s'_{R'}  \right) \mcom
\]
if $t$ is odd. 
\end{definition}
\begin{figure}[h]
    \centering
    \includegraphics[width=0.75\linewidth]{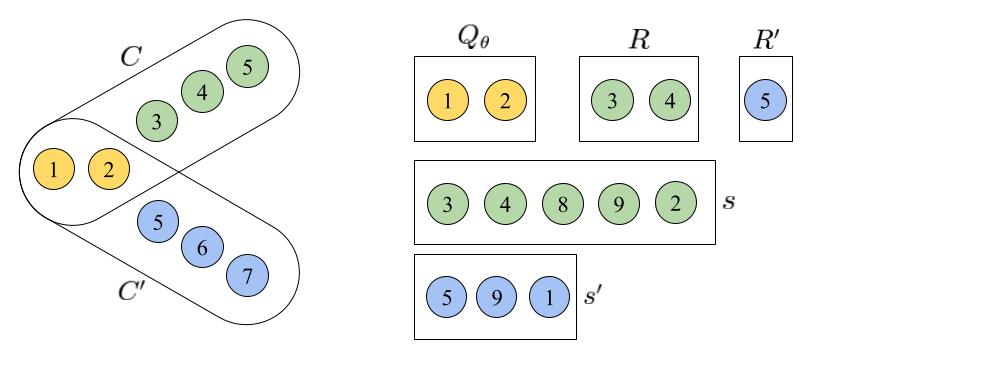}
    \caption{Example of a contributing monomial to $\Deg(s,s')$. Here, $q = 5, t = 2$. As in \cref{fig:KikuchiGraph}, vertices of the form $(x, 1)$ are colored green, while vertices of the form $(x, 2)$ are colored blue. The elements of $Q_\theta$ are colored yellow. Note that $q - t = 3$ is odd, and $S$ indicated by $(s,s')$ contains $2 = \left\lceil\frac{q - t}{2}\right\rceil$ green elements from $\left(C\setminus Q_\theta\right)\times\{1\}$, and $1 = \left\lfloor\frac{q - t}{2}\right\rfloor$ blue element from $\left(C'\setminus Q_\theta\right)\times\{2\}$.}
    \label{fig:Deg(ss')}
\end{figure}

\begin{remark}[Understanding $\Deg(s,s')$]
In either case, note that $\Deg(s,s')$ is a polynomial of degree $q-t$. For any $(s,s')$ that indicates a subset of size $\ell$ in $[n]\times [2]$,  $\Deg(s,s')$ upper bounds the degree of the vertex $S$ in $K_{i,t}$. This is because for each pair $(C,C')$ that contributes a $D$-matching to $K_{i,t}$, there is a corresponding monomial that evaluates to at least $1$ in $\Deg(s,s')$. We note that in $\Deg(s,s')$, we have contributions from all $C,C'$ such that $C \in \cH_i$ and $C,C'$ are in some piece $\cH^{(t)}_{\theta}$ of the partition promised by the approximate strong regularity of $\cH$. In particular, $\Deg(s,s')$ contains contributions from pairs $(C,C')$ such that $C' \in \cH_j$ for some $j \in L$ even though such a pair does not contribute any edges to $K_{i,t}$. This, of course, can only hurt us in our upper bounds on the degree. 
\end{remark}
To understand the distribution of degrees of the vertices $S$ of $K_{i,t}$, we must analyze the distribution of $\Deg(s,s')$ as $(s,s')$ vary over indicators of subsets of size exactly $\ell$. As in~\cite{GuruswamiKM22}, we use a coupling argument to instead work with the $\ell/n$-biased product distribution over $\zo^{2n}$.

\begin{fact}[Similar to Lemma 6.20 of~\cite{GuruswamiKM22}] \label{fact:sscoupling}
Let $U$ be the $\frac{\ell}{n}$-biased product distribution on $(s,s')$. Let $U'$ be the uniform distribution on $(s,s') \in \zo^{2n}$ indicating sets of size exactly $\ell$. Then for any $\lambda$, we have

\[\Pr_{(s,s') \sim U'} [\Deg(s,s')\leq \lambda]\geq \Pr_{(s,s') \sim U} \left[\Deg(s,s') \leq \lambda\right] - e^{-\ell / 4}.\]

\end{fact}

Our plan now is to apply the polynomial concentration inequality to understand the behavior of the polynomial $\Deg(s,s')$. The next lemma bounds the expectations of partial derivatives of $\Deg(s,s')$ towards this application. 

\begin{lemma}\label{lem:derivativebounds}
Suppose $(s,s') \in \zo^{2n}$ has the $\frac{\ell}{n}$-biased product distribution. Let $Z_1$ be any multiset of elements from $s_1, s_2, \ldots, s_n$, and let $Z_2$ be any multiset of elements from $s'_1, s'_2, \ldots, s'_n$. Then, 
\[
\E\left[\left(\prod_{ j \in Z_2}\partial_{s'_j} \right)\Deg(s,s') \right]\leq 2^q \left(\frac{\ell}{n}\right)^{q-t-|Z_2|} d_{|Z_2|} \mcom
\]
and,
\[
\E\left[\left(\prod_{i \in Z_1} \partial_{s_i} \prod_{ j \in Z_2}\partial_{s'_j} \right)\Deg(s,s') \right]\leq 2^q \left(\frac{\ell}{n}\right)^{q-t-|Z_1|-|Z_2|} d_{|Z_2|+t}\mper
\]
In particular, if $\gamma_i = \log_n d_i$ for all $i \in [q]$ and $\gamma_1 \geq 1-2/q$, and our hypergraph is $t$-approximately strongly regular,
\[
\E\left[\left(\prod_{i \in Z_1} \partial_{s_i} \prod_{ j \in Z_2}\partial_{s'_j} \right)\Deg(s,s') \right]\leq 2^q \left(\frac{\ell}{n}\right)^{q-t} \frac{nd_t}{(\log n)^{|Z_1| + |Z_2|}} \mper
\]

\end{lemma}
\begin{proof}
Without loss of generality, we will assume that $Z_1, Z_2$ are sets (since differentiating twice with respect to the same variable immediately makes the polynomial vanish). 

When $Z_1 = Z_2= \emptyset$, we observe that 
\[
\E\left[\Deg(s,s')\right] \leq 2^{2q} \left(\frac{\ell}{n}\right)^{q-t} n d_t \mper
\numberthis \label{eq:expected-Deg-bound}
\]
To see why, observe that the number of $(C,C')$ that contribute monomials to $\Deg(s,s')$ can be upper bounded by $n d_t$ since there are $\leq n$ possible choices of $C \in \cH_i$. Given $C \in \cH_i$, there is at most one $Q_{\theta} \subseteq C$ such that $C \in \cH^{(t)}_{\theta}$, and any $Q_{\theta}$ appears in at most $d_t$ different $C'$ in $\cH$. Further, each pair $(C,C')$ contributes at most $2^{2q}$ monomials corresponding $R,R'$ that range over subsets of $C,C'$ respectively. Finally, each monomial has degree exactly $q-t$ and thus has expectation $(\ell/n)^{q-t}$ under the $\ell/n$-biased product distribution.  

Let us start with the case when $Z_1 = \emptyset$ and $Z_2 \neq \emptyset$. Clearly, $|Z_2|\leq \lceil (q-t)/2 \rceil$. Our overall strategy is similar to the above -- we will count the number of pairs $(C,C')$ that can contribute a non-zero monomial to $\big(\prod_{j \in Z_2}\partial_{s'_j} \big) \Deg(s,s')$. First, each $(C,C')$ contributes at most $2^{2q}$ different monomials, and each monomial has an expectation of $(\ell/n)^{q-t-|Z_2|}$. To count the number of pairs, the only monomials that survive after taking the partial derivatives are ones where $C' \supseteq Z_2$. Thus, the number of possible $C'$ is at most $d_{|Z_2|}$. Given $C'$, we have a unique $\theta$ such that $C' \in \cH^{(t)}_{\theta}$, and given $Q_{\theta}$, there is at most one $C \in \cH_i$ (since $\cH_i$ is a matching) such that $C \supseteq Q_{\theta}$. Thus, we obtain an upper bound of $d_{|Z_2   |} 2^q (\ell/n)^{q-t-|Z_2|}$ in this case.  

Now, $(\ell/n)^{q-t-|Z_2|} d_{|Z_2|} \leq (\ell/n)^{q-t}n d_t\cdot(\log n)^{-|Z_2|}$ if and only if $d_{|Z_2|} \leq (\ell/n)^{|Z_2|} n d_t\cdot(\log n)^{-|Z_2|}$. Since $\ell = n^{1-2/q}\log n$, this is true if and only if $d_{|Z_2|} \leq n^{1-2 |Z_2|/q}d_t$ or $\gamma_{|Z_2|} \leq 1-2/q + \gamma_t$, a condition that is satisfied since $t$ satisfies the first condition in \Cref{def:goodindex}. 

Finally, consider the case when $Z_1$ is non-empty. Since the rest of the argument is similar, let's focus on counting the number of pairs $(C,C')$ such that $C \supseteq Z_1$. Clearly, there is at most one $C \in \cH_i$ that can contain $Z_1$. Given $C$, we know the unique $\cH^{(t)}_{\theta}$ that contains it and thus know $Q_{\theta}$. So we are looking for $C'$ that contain $Q_{\theta} \cup Z_2$ -- a set of size $t+|Z_2|$. Thus, the number of such $C'$ is at most $d_{t+|Z_2|}$. 

Now, $(\ell/n)^{q-t-|Z_1|-|Z_2|} d_{t+|Z_2|} \leq (\ell/n)^{q-t} n d_t\cdot(\log n)^{-(|Z_1| + |Z_2|)}$ for $\ell = n^{1-2/q}\log n$ if and only if $d_{t+|Z_2|} \leq n^{1-2(|Z_1|+|Z_2|)/q} d_t$. If $|Z_2|\neq \lceil(q-t)/2\rceil,$ we set $|Z_1|$ to its maximum of $\lceil (q-t)/2 \rceil$ to see this is satisfied if $d_{t+|Z_2|} \leq n^{-2(|Z_2|-t)/q} n^{-(1/q) \1 (q-t \text{ is odd})} d_t,$ which is true by the second condition in \Cref{def:goodindex}. Otherwise, we let $g = |Z_2|-1$ and $f=|Z_1|+1$, and since the degrees are decreasing, we get $d_{t+|Z_2|}\leq d_{t+g}
\leq n^{1-2(f+g)/q}d_t=n^{1-2(|Z_1|+|Z_2|)/q} d_t.$  \qedhere

\end{proof}

\begin{proof}[Proof of \Cref{lem:kikuchiregularity}]
Given the degree sequence $d_1 \geq d_2 \geq \cdots \geq d_q$, set $\gamma_i = \log_n d_i$ for every $1 \leq i \leq q$. Then, notice that the $\gamma_i$'s satisfy the hypothesis of \Cref{lem:good-index-exists} since we set $k$ large enough so that $d_1 \geq |\cH| /n \geq n^{1-2/q}$. 
Since our hypergraph is $t$-approximately strongly regular, 
property (\ref{item:gamma-t-lower-bound}) of \Cref{def:goodindex} implies that $d_t \geq d_1 n^{-2(t-1)/q}$ if $t < q/2$ and $d_t \geq d_1 n^{-1+2/q}$ if $t > q/2$,
i.e., the desired lower bound for $d_t$ stated in \Cref{lem:kikuchiregularity}.
Note that using the assumption $d_1 \geq \ell = n^{1-2/q}$, it follows that the upper bound on $\E[\Deg(s,s')]$ in \eqref{eq:expected-Deg-bound} is at least $n^{-\frac{2}{q}(q-t)} nd_t \geq 1$.

We apply \Cref{fact:ss} to $\Deg(s,s')$ over the $(\ell/n)$-biased product distribution over $\zo^n$. Given \Cref{lem:derivativebounds}, we immediately obtain that for large enough absolute constant $W>0$,
\[
\Pr\left[\Deg(s,s') \geq  \left(\frac{\ell}{n}\right)^{q-t} n d_t W ^q \right] \leq n^{-100q} \mper
\]
Combining with \Cref{fact:sscoupling} finishes the proof.
\end{proof}

%% file: new_bucketing_lemma.tex
\section{Approximate Strong Regularity Decomposition}
\label{decompsec}
In this section, we prove \cref{lem:new_bucketing_arpon_version}. We use the following technical (but rather crucial!) result, which we prove in \cref{sec:numberslemma}, to show that a good index as defined in \cref{def:goodindex} always exists.

\begin{lemma}
\torestate{
\label{lem:good-index-exists}
    Let $q\geq 3$ be an odd integer, and let $\gamma_1\geq\gamma_2\geq\cdots\geq\gamma_{q - 1}\geq\gamma_q\geq 0$ be real numbers. Then there exists $t\in[q]$ satisfying the following properties:
    \begin{enumerate}[(1)]
        \item $\gamma_r\leq\gamma_t + 1 - \frac{2r}{q}$ for all $1\leq r\leq\left\lceil\frac{q - t}{2}\right\rceil$.
        \label{item:small-i}
        \item $\gamma_r\leq\gamma_t - \frac{2(r - t)}{q} + \frac{1}{q}\left(t - \mathbf{1}(t\text{ even})\right)$ for all $t < r\leq\left\lfloor\frac{q + t}{2}\right\rfloor$.
        \label{item:large-i}
        \item If $t < \frac{q}{2}$, then $\gamma_t- \big(1 - \frac{2t}{q}\big)\geq \gamma_1 - \big(1-\frac{2}{q}\big)$. If $t > \frac{q}{2}$, then $\gamma_t > \gamma_1 - \big(1-\frac{2}{q}\big)$.
        \label{item:gamma-t-lower-bound}
    \end{enumerate}
    Thus, by taking $\gamma_i=\log_n d_i$ for all $i\in [q],$ every $q$-uniform hypergraph contains a good index.
}
\end{lemma}
Now, we restate the approximate strong regularity decomposition lemma. We refer the reader to \Cref{sec:decomposition-overview} for an overview.

\restatelemma{lem:new_bucketing_arpon_version}

We prove \Cref{lem:new_bucketing_arpon_version} by describing an algorithm to produce the required partition, which uses as a subroutine the following \Cref{alg:newstrongreg}.
\begin{mdframed}
    \begin{algorithm}
    \label{alg:newstrongreg}    
    \mbox{}
    \begin{description}
        \item[Input:] A hypergraph $\cH$ on vertices $[n]$, and an integer $t\geq 1$.
        \item[Output:] Two hypergraphs $(\cH^{\mathrm{regular}}, \cH^{\mathrm{residual}})$ where $\cH = \cH^{\mathrm{regular}}\sqcup\cH^{\mathrm{residual}}$.
        \item[Operation:] \mbox{}
        \begin{enumerate}
        \item Set $\cH^{\mathrm{regular}}\coloneqq \emptyset, \cH^{\mathrm{residual}}\coloneqq \cH, D\coloneqq d_{\cH, t}$.
        \item While $d_{\cH^{\mathrm{residual}}, t}\geq D/2$:
        \begin{enumerate}
            \item \label{newsrg:peelingstep} Consider $Q\in\binom{V(\cH)}{t}$ such that $\big|\cH^{\mathrm{residual}}_{|Q}\big|$ is the maximum. Then perform $\cH^{\mathrm{regular}}\leftarrow\cH^{\mathrm{regular}}\cup\cH^{\mathrm{residual}}_{|Q}, \cH^{\mathrm{residual}}\leftarrow\cH^{\mathrm{residual}}\setminus\cH^{\mathrm{residual}}_{|Q}$. 
        \end{enumerate}
        \item  Output $(\cH^{\mathrm{regular}}, \cH^{\mathrm{residual}})$.
\end{enumerate}
    \end{description}
    \end{algorithm}
\end{mdframed}
\begin{remark}
\label{newSRDremarks}
Note that the output of this procedure depends on which $Q$ we choose in \cref{newsrg:peelingstep} if there are multiple $Q$ such that $|\cH^{\mathrm{residual}}_{|Q}|$ is maximal. Any arbitrary choice of $Q$ is fine for our application.
\end{remark}
In \cref{alg:newstrongreg}, observe that $d_{\cH^{\mathrm{regular}}, r}\leq d_{\cH, r}$ for all $r\in[q]$, since $\cH^{\mathrm{regular}}\subseteq\cH$. Also observe that if $d_{\cH, t} \geq 2$, then $d_{\cH^{\mathrm{residual}}, t} < d_{\cH, t}/2 < d_{\cH, t} = d_{\cH^{\mathrm{regular}}, t}$.

Our main decomposition algorithm iteratively applies \Cref{alg:newstrongreg}.
\begin{mdframed}
    \begin{algorithm}[Approximate Strong Regularity Decomposition]\mbox{}
    \label{alg:caterpillar}
    \begin{description}
        \item [Input:] $\cH$, a $q$-uniform hypergraph on $[n]$.
        \item [Output:] $\mathcal{T}$, a collection of hypergraphs partitioning $\cH$.
        \item[Operation:] \mbox{}
        \begin{enumerate}
            \item Define $\cH^{\mathrm{curr}}\coloneqq \cH, \mathcal{T}\coloneqq ()$.
            \item While $\cH^{\mathrm{curr}}\neq \emptyset$:
            \begin{enumerate}
                \item Write $\gamma_i\coloneqq \log_n d_{\cH^{\mathrm{curr}}, i}$ for all $i\in[q]$, and let $t\in[q]$ be the index returned by invoking \cref{lem:good-index-exists} on $\{\gamma_i\}_{i\in[q]}$.
                \item \label{item:bisectionstep} Apply \cref{alg:newstrongreg} on $\cH^{\mathrm{curr}}$ and $t$ to get hypergraphs $(\cH^{\mathrm{curr}})^{\mathrm{regular}}, (\cH^{\mathrm{curr}})^{\mathrm{residual}}$. Write $\cH^{\mathrm{curr}}\leftarrow(\cH^{\mathrm{curr}})^{\mathrm{residual}}, \mathcal{T}\leftarrow(\mathcal{T}, (\cH^{\mathrm{curr}})^{\mathrm{regular}})$.
            \end{enumerate}
            \item Return $\mathcal{T}$.
        \end{enumerate}
    \end{description}
    \end{algorithm}
\end{mdframed}
In the above algorithm, $(\mathcal{T}, \cH')$ stands for appending the hypergraph $\cH'$ at the end of the tuple $\mathcal{T}$.\\
We record some small propositions about the above algorithm.
\begin{proposition}
\label{prop:strongregprop}
    Suppose $\mathcal{T}$ is the output of \cref{alg:caterpillar} on a hypergraph $\cH$. Let $\cH'$ be any but the last entry of $\mathcal{T}$.  Let $t$ be the index in the invocation of \cref{item:bisectionstep} which produced $\cH'$. Then:
    \begin{enumerate}
        \item $\cH'$ is $t$-approximately strongly regular.
        \item Write $\gamma'_i\coloneqq \log_n d_{\cH', t}$. Then $\{\gamma'_i\}$ satisfies the conditions of \cref{lem:good-index-exists} with the ``good index'' being $t$.
    \end{enumerate}
\end{proposition}
\begin{proof}
    Consider any $C\in\mathcal{H}'$. During the iteration of \Cref{alg:newstrongreg} on inputs $(\cH^{\mathrm{curr}},t)$ 
    where $C$ was added to $\mathcal{H}'$, suppose that $\mathcal{H}^{\mathrm{curr}}=\mathcal{H}''$. Note that $d_{\cH', t} = d_{\cH'', t}$.
    Then there exists $Q\in{V(\mathcal{H})\choose t}$ such that $Q\subseteq C\in \mathcal{H}''_{|Q}$  and $\big|\mathcal{H}'_{|Q}\big|\geq d_{\cH', t}/2$, and $\big|\mathcal{H}'_{|Q}\big|\leq d_{\cH', t}$ holds by definition. Thus $\cH^j$ can  be written as the union of these $\cH'_{|Q}$'s, where we enumerate arbitrarily by $[p_t]$.
    
    The second statement follows by noting that $\gamma'_t = \gamma_t$, and $\gamma'_i\leq\gamma_i$.
\end{proof}
\begin{proposition}
\label{lem:caterpillarfastfinish}
    For the above algorithm, $|\mathcal{T}|\leq q\lceil\log_2|\cH|\rceil + 1$.
\end{proposition}
\begin{proof}
    Write $\cH' = \cH^{\mathrm{curr}}$ at any step of this algorithm, and let $\cH''$ be $ \cH^{\mathrm{curr}}$ after we apply \cref{item:bisectionstep}. Note that for some $t\in[q]$, we had that $d_{\cH'', t} < d_{\cH', t}/2$. Now, assume for the sake of contradiction that $|\mathcal{T}|> q\lceil\log_2|\cH|\rceil + 1$. Then by the pigeonhole principle, there must be some $t\in[q]$ which was the index invoked in \cref{item:bisectionstep} more than $\lceil\log_2|\cH|\rceil $ times. Now, let $\cH'''$ be the last entry of $\mathcal{T}$. Then we have $d_{\cH''', t} < d_{\cH, t}/2^{\lceil\log_2|\cH|\rceil } \leq 1$, where the last inequality follows since $d_{\cH, t}\leq|\cH|$. But $d_{\cH''', t} < 1$ is a contradiction, and the result follows.
\end{proof}
We can finally prove \cref{lem:new_bucketing_arpon_version}.
\begin{proof}[Proof of \cref{lem:new_bucketing_arpon_version}]
    Run \cref{alg:caterpillar} on $\cH$ to obtain $\mathcal{T}$. Write $T\coloneqq |\mathcal{T}|$. By \cref{lem:caterpillarfastfinish}, we know that $T\leq q\lceil\log_2|\cH|\rceil + 1$.
    
    Write $\{\cH^\pi:\pi\in P\}$\footnote{$P$ is just an indexing set, and $\pi$ is the corresponding indexing variable.} to be the hypergraphs in $\mathcal{T}$ with size $\geq\frac{\eta|\cH|}{T}$. Note that the number of hyperedges of $\cH$ not contained in $\{\cH^\pi:\pi\in P\}$ has size $ < \frac{\eta|\cH|}{T}\cdot T = \eta|\cH|$. \cref{lem:new_bucketing_arpon_version} now follows from \cref{prop:strongregprop,lem:caterpillarfastfinish}.
\end{proof}
\input{numbers-lemma}

%% file: numbers-lemma.tex
\subsection{Proof of \texorpdfstring{\Cref{lem:good-index-exists}}{Lemma~\ref{lem:good-index-exists}}: a good index exists}
\label{sec:numberslemma}

In this section, we prove \Cref{lem:good-index-exists}.
We first show that for inequality (\ref{item:small-i}) in \Cref{lem:good-index-exists}, we may assume $r < t$.

\begin{lemma} \label{lem:2-implies-1}
    In \Cref{lem:good-index-exists}, when $t < r \leq \ceil*{\frac{q-t}{2}}$, we have (\ref{item:large-i}) $\implies$ (\ref{item:small-i}).
\end{lemma}
\begin{proof}
    Since $t < \ceil*{\frac{q-t}{2}}\leq \frac{q-t}{2}$, we have $3t \leq q$.
    Then, $-\frac{2(r-t)}{q} + \frac{1}{q}(t-\1(t \text{ even})) \leq \frac{1}{q}(3t-2r) \leq 1 - \frac{2r}{q}$.
    This means that (\ref{item:large-i}) is a tighter upper bound than (\ref{item:small-i}).
\end{proof}

We also show the following useful lemma.

\begin{lemma} \label{lem:no-larger-violation}
    If $t > q/2$, then inequality (\ref{item:large-i}) is always satisfied.
\end{lemma}
\begin{proof}
    This follows from the descending properties of the $\gamma_t$'s.
    For any $r$ such that $t < r \leq \floor*{\frac{q+t}{2}}$, we have
    $-\frac{2(r-t)}{q} + \frac{1}{q}(t - \1(t\text{ even})) \geq -\frac{2}{q} \cdot \frac{q-t-\1(t\text{ even})}{2} + \frac{1}{q}(t - \1(t\text{ even})) = \frac{1}{q}(2t-q) > 0$.
    Thus, since $r > t$, we have $\gamma_r \leq \gamma_t$, and inequality (\ref{item:large-i}) holds.
\end{proof}

\begin{proof}[Proof of \Cref{lem:good-index-exists}]
    We use the following algorithm to find the desired index.
    \begin{enumerate}
        \item Set $t_0 < q/2$ to be the index that maximizes $\gamma_{t_0} + \frac{2t_0}{q}$.
        \item If there is some $t > q/2$ such that that $\gamma_t$ violates inequality (\ref{item:large-i}) for $\gamma_{t_0}$, then return the $t$ that maximizes $\gamma_t + \frac{2t}{q}$ over the violating indices.
    \end{enumerate}
    First, we claim that for all $r < t_0$, $\gamma_r$ satisfies inequality (\ref{item:small-i}) for $\gamma_{t_0}$.
    This follows from the fact that $t_0$ maximizes $\gamma_{t_0} + \frac{2t_0}{q}$, thus we have
    \begin{align*}
        \gamma_r + \frac{2r}{q} \leq \gamma_{t_0} + \frac{2t_0}{q} \implies \gamma_r \leq \gamma_{t_0} + \frac{2t_0}{q} - \frac{2r}{q} < \gamma_{t_0} + 1 - \frac{2r}{q} \mcom
    \end{align*}
    since $t_0 < q/2$.
    Moreover, we show that no index $r \in [t_0+1, \frac{q-1}{2}]$ violates inequality (\ref{item:large-i}) for $\gamma_{t_0}$.
    Again, we use $\gamma_r \leq \gamma_{t_0} - \frac{2(r-t_0)}{q}$, which immediately implies that inequality (\ref{item:large-i}) is satisfied.
    Therefore, if there is an index $t$ that violates inequality (\ref{item:large-i}) for $\gamma_{t_0}$, we must have $t > q/2$.

    Next, \Cref{lem:no-larger-violation} states that inequality (\ref{item:large-i}) is satisfied for any $r > t$ since $t > q/2$.
    Thus, we now verify that inequality (\ref{item:small-i}) is satisfied for all $r \leq \ceil*{\frac{q-t}{2}}$.
    Since $t$ violates inequality (\ref{item:large-i}) for $\gamma_{t_0}$, we must have $t \leq \floor*{\frac{q+t_0}{2}} = \frac{q+t_0-\1(t_0\text{ even})}{2}$ (as $q$ is odd)
    and $\gamma_t > \gamma_{t_0} - \frac{2(t-t_0)}{q}+ \frac{1}{q}(t_0 - \1(t_0\text{ even}))$, which imply that $\gamma_t > \gamma_{t_0} - 1 + \frac{2t_0}{q}$.
    Combined with $\gamma_r \leq \gamma_{t_0} + \frac{2(t_0-r)}{q}$, we get
    \begin{align*}
        \gamma_r \leq \gamma_t + 1 - \frac{2t_0}{q} + \frac{2(t_0-r)}{q}
        = \gamma_t + 1 - \frac{2r}{q} \mcom
    \end{align*}
    thus satisfying inequality (\ref{item:small-i}).

    Finally, we show lower bounds on $\gamma_{t_0}$ or $\gamma_t$.
    First, we have $\gamma_{t_0} + \frac{2t_0}{q} \geq \gamma_1 + \frac{2}{q}$ since $t_0 < q/2$. But $\gamma_{t_0} + \frac{2t_0}{q} \geq \gamma_1 + \frac{2}{q}\implies\gamma_{t_0}\geq 1 - \frac{2t_0}{q} + \big(\gamma_1 - \big(1 - \frac{2}{q}\big)\big)$, as desired. For $\gamma_t$ (if we reach step 2 of the algorithm), we have that $\gamma_t > \gamma_{t_0} - 1 + \frac{2t_0}{q} \geq \gamma_1 - \big(1 - \frac{2}{q}\big)$.
    This completes the proof.
\end{proof}